\documentclass{amsart}
\usepackage{amsmath}
\usepackage{amssymb,amsthm}
\usepackage{graphicx}
\numberwithin{equation}{section}
\newtheorem{thm}{Theorem}[section]
\newtheorem{rem}[thm]{Remark}

\newtheorem{lem}[thm]{Lemma}

\newtheorem{dfn}[thm]{Definition}
\newtheorem{pro}[thm]{Proposition}
\newtheorem{as}[thm]{Assumption}

\title[Scattering matrices and generalized Fourier transforms]{Scattering matrices and generalized Fourier transforms in long-range $N$-body problems}
\author{Sohei Ashida}
\begin{document}
\maketitle

\begin{abstract}
We give a definition of scattering matrices based on the asymptotic behaviors of generalized eigenfunctions and show that these scattering matrices are equivalent to the ones defined by wave-operator approach in long-range $N$-body problems. We also define generalized Fourier transforms by the asymptotic behaviors of outgoing solutions to nonhomogeneous equations and show that they are equivalent to the definition using wave operators. We also prove that the adjoint operators of the generalized Fourier transforms are given by Poisson operators.
\end{abstract}

\section{Introduction}\label{firstsec}
Scattering matrices play an important role in the study of long-time asymptotic behaviors of the solutions to Schr\"odinger equations. Scattering matrices are defined in two different ways. In the time dependent viewpoint, the scattering matrices are defined using wave operators and the Fourier transforms. On the other hand, in the stationary viewpoint, they are defined using the asymptotic behaviors of generalized eigenfunctions at infinity. In this paper we prove that both the definitions are equivalent in long-rang $N$-body problems. We also give a definition of the generalized Fourier transforms using the asymptotic behaviors of outgoing solutions to nonhomogeneous equations. We prove that they are equivalent to the ones using wave operators, and that tthere adjoint operators are given by Poisson operators.

Before we consider the $N$-body problems, it is instructive to recall the results for two-body problems, that is, the cases with decaying potentials. Let the potential $V(x)\in C^{\infty}(\mathbb R^n),\ n\in \mathbb N$ satisfy
$$\lvert \partial^{\gamma}V(x)\rvert=\mathcal O(\lvert x\rvert^{-\mu-\lvert\gamma\rvert}),$$
for $\mu>0$ as $\lvert x\rvert\to\infty$.
For short-range potentials, that is, when $\mu>1$, as time $t$ tends to $\pm\infty$ the asymptotic behaviors of the solutions $e^{-it\tilde H}\psi$, $\tilde H:=-\Delta+V$, $\psi\in L^2(\mathbb R^n)$ to the Schr\"odinger equations are given by the free evolution $e^{it\Delta}\psi_{\pm}$, $\psi_{\pm}\in L^2(\mathbb R^n)$, where $\Delta$ is the Laplacian on $\mathbb R^n$. In other words,
$$\lVert e^{-it\tilde H}\psi-e^{it\Delta}\psi_{\pm}\rVert\to0,$$
as $t\to\pm\infty$. On the contrary, for any $\psi_{\pm}\in L^2(\mathbb R^n)$ there exists $\psi\in L^2(\mathbb R^n)$ such that
$$\lVert e^{it\Delta}\psi_{\pm}-e^{-it\tilde H}\psi\rVert\to0,$$
as $t\to\pm\infty$.

The wave operators $W_{\pm}: L^2(\mathbb R^n)\to L^2(\mathbb R^n)$ are defined by $W_{\pm}\psi_{\pm}:=\psi$. The wave operators $W_{\pm}$ are partial isometries from $L^2(\mathbb R^n)$ to $\mathcal H_{ac}(\tilde H)$, where $\mathcal H_{ac}(\tilde H)$ is the absolutely continuous subspace of $\tilde H$.

The scattering operator $S$ is defined as a map $S\psi_-:=\psi_+$. Let $\mathbf F_0$ be the Fourier transform. Then, $\hat S:=\mathbf F_0S\mathbf F_0^{*}$ commutes with any bounded Borel functions of $\lvert \xi\rvert^2$, and therefore, there exist $\hat S(\lambda)\in \mathcal L(L^2(\mathbb S^{n-1}),L^2(\mathbb S^{n-1}))$, $a.e.\, \lambda>0$ such that
$$(\hat Sf)(\lambda,\omega)=(\hat S(\lambda)f(\lambda))(\omega),\ \xi=\sqrt{\lambda}\omega,\ \omega\in \mathbb S^{n-1},$$
$a.e.\, \lambda>0$ for any $f(\xi)\in L^2(\mathbb R^n)$ (see e.g. Reed-Simon \cite{RS}). Here $(f(\lambda))(\omega):=f(\lambda,\omega)$.
The operators $\hat S(\lambda)$ are called scattering matrices. Thus the scattering matrices give the correspondence between the data as $t\to \pm\infty$.

On the other hand, there is another definition of the scattering matrix which is known to be equivalent to the definition as above. The another definition is based on the asymptotic behaviors of the generalized eigenfunctions at infinity (see Melrose \cite{Me} and Yafaev \cite{Ya}). The generalized eigenfunctions are solutions $u$ to $(\tilde H-\lambda)u=0$, $\lambda\in \mathbb C$, in the sense of distributions.
For short range potentials it is known that for any $a_-\in C^{\infty}(\mathbb S^{n-1})$ there exists a unique generalized eigenfunction $u$ of $\tilde H$ of the form
\begin{equation}\label{myeq1.1}
\lvert x\rvert^{-(n-1)/2}(e^{-i\sqrt{\lambda}\lvert x\rvert}a_-(\hat x)+e^{i\sqrt{\lambda}\lvert x\rvert}a_+(\hat x))+\mathcal O(\lvert x\rvert^{-(n+1)/2}),
\end{equation}
as $\lvert x\rvert \to \infty$ with $a_+\in C^{\infty}(\mathbb S^{n-1})$ uniquely determined by $a_-$, where $\hat x:=x/\lvert x\rvert$.
The scattering matrix $\Sigma(\lambda)$ is defined as the map
$$\Sigma(\lambda) : a_-\mapsto a_+.$$

The scattering matrices $\hat S(\lambda)$ and $\Sigma(\lambda)$ are equivalent in the sense that the following equation holds (see e.g. Reed-Simon \cite{RS} and Melrose \cite{Me}).
\begin{equation}\label{myeq1.2}
\hat S(\lambda)=i^{n-1}\Sigma(\lambda).
\end{equation}

In the case of long-range two-body problems, that is, when $1\geq \mu>0$, we need to modify both the definitions of $\hat S(\lambda)$ and $\Sigma(\lambda)$. In the case of long-range potentials  the free evolution $e^{it\Delta}\psi_{\pm}$, $\psi_{\pm}\in L^2(\mathbb R^n)$ is replaced by $e^{-iS_{\pm}(D,t)}\psi_{\pm}$, where $S_{\pm}(\xi,t)$ are solutions to the Hamilton-Jacobi equation
$$\frac{\partial S_{\pm}}{\partial t}(\xi,t)=\lvert \xi\rvert^2+V(\nabla_{\xi}S_{\pm}(\xi,t)),$$
and
$$e^{-iS_{\pm}(D,t)}\psi:=\mathbf F_0^{*}(e^{-iS_{\pm}(\xi,t)}(\mathbf F_0\psi)(\xi)).$$

On the other hand $e^{\pm i\sqrt{\lambda}\vert x\rvert}$ in \eqref{myeq1.1} is replace by $e^{\pm iK(x,\lambda)}$, where $K(x,\lambda)$ is a suitably chosen solution to the eikonal equation
$$\lvert \nabla_xK(x,\lambda)\rvert^2+V(x)=\lambda.$$

Then, we have the relation \eqref{myeq1.2} even in the long-range case (see G\^atel-Yafaev \cite{GY}).

We now turn to the $N$-body problems. We consider the generalized $N$-body Schr\"odinger operators. The $N$-body Schr\"odinger operator is a special case of the generalized $N$-body Schr\"odinger operators.

Set $X:=\mathbb R^n,\ n\in \mathbb N$, and let $\{X_a\subset X: a\in\mathcal A\}$ be a finite family of linear subspaces of $X$ which is closed under intersections, and includes $X$ and $X_{a_{\max}}:=\{0\}$. We endow $\mathcal A$ with a semi-lattice structure by
$$a\leq b\ \mathrm{if}\ X_a\supset X_b.$$
We denote the orthogonal complement of $X_a$ by $X^a$. We denote by $\Pi^a$ and $\Pi_a$ the orthogonal projections of $X$ onto $X^a$ and $X_a$ respectively. We use the same notations $\Pi^a$ and $\Pi_a$ for the corresponding orthogonal projections of the dual space of $X$. We define for all $x\in X$, $x_a=\Pi_a x$ and $x^a=\Pi^a x$. If $a\leq b$, we define
\begin{equation}\label{myeq1.2.0.1}
x_a^b:=\pi_a\pi^b x.
\end{equation}
We also define $\nabla_a=\Pi_a\nabla$ and $\nabla^a=\Pi^a\nabla$. The operators $-\Delta_a$ and $-\Delta^a$ denote the Laplacian in $X_a$ and $X^a$ respectively. 

We define $C_a:=X_a\cap\mathbb S^{n_a-1}$. Thus $C_a$ is a sphere of dimension $n_a-1$.  We also define the singular part of $C_a$ by
$$C_{a,\mathrm{sing}}:=\bigcup_{b\nleq a}(C_b\cap C_a).$$
 $C_{a,\mathrm{sing}}$ corresponds to the directions in which the particles collide. The regular part $C_a'$ is the complement of $C_{a,\mathrm{sing}}$:
 $$C_a':=C_a\setminus C_{a,\mathrm{sing}}.$$

A generalized  $N$-body Schr\"odinger operator is an operator of the form
\begin{equation}\label{myeq1.2.1}
H:=-\Delta+\sum_{a\in \mathcal A}V_a(x^a),
\end{equation}
where $\Delta$ is the Laplacian in $X$ and $V_a$ are real-valued functions on $X^a$ satisfying the following condition.
There exists $\mu>0$ such that  for any $a\in \mathcal A$, $V_a(x^a)=V_a^s(x^a)+V_a^l(x^a)$, where
\begin{itemize}
\item[(1)]$V_a^s(x^a)$ is compactly supported and $V_a^s(-\Delta^a+1)^{-1}$ is compact.
\item[(2)]$V_a^l(x^a)\in C^{\infty}(X^a)$ and for any $\gamma\in \mathbb N^{n_a}$
\begin{equation}\label{myeq1.2.1.1}
\lvert\partial^{\gamma} V_a(x^a)\rvert=\mathcal(\lvert x^a\rvert^{-\mu-\lvert \gamma\rvert}),
\end{equation}
where $n_a:=\mathrm{dim}\, X_a$.
\end{itemize}
Then, $H$ is a self-adjoint operator on $L^2(X)$.

We also define the operators $H^a$ as
\begin{equation}\label{myeq1.2.2}
H^a:=-\Delta^a+\sum_{b\leq a}V_b(x^b).
\end{equation}
The set of thresholds of a subsystem $a\in \mathcal A$ is defined as
\begin{equation*}\label{myeq1.2.3}
\mathcal T^a:=\bigcup_{b<a}\sigma_{\mathrm{pp}}(H^b),
\end{equation*}
where $\sigma_{pp}(A)$ is the set of eigenvalues of $A$, and $b<a$ means $b\leq a$ and $b\neq a$. We also set
\begin{equation*}\label{myeq1.2.4}
\mathcal T(H):=\mathcal T^{a_{\max}}.
\end{equation*}

We label the eigenvalues of $H^a$ counted with multiplicities, by integers $m$, and we call the pairs $\alpha=(a,m)$ channels. We denote the eigenvalue of the channel $\alpha$ and the corresponding normalize eigenfunction by $E_{\alpha}$ and $u_{\alpha}$ respectively. We say that a channel $\alpha$ is a non-threshold channel if $E_{\alpha}\notin \mathcal T^a$. When $\alpha$ is a non-threshold channel, the eigenfunction $u_{\alpha}$ is a Schwartz funcition (see Froese-Herbst \cite{FH}).

When $\mu>1$ and $\alpha$ is a non-threshold channel, the following strong limit exists:
$$W^{\pm}_{\alpha}:=s-\lim_{t\to\pm\infty}e^{itH}J_{\alpha}e^{-it(\Delta_a+E_{\alpha})},$$
where $(J_{\alpha}v)(x):=u_{\alpha}(x^a)v(x_a)$.

When $\alpha$ and $\beta$ are non-threshold channels corresponding to $a\in \mathcal A$ and $b\in \mathcal A$ respectively, the scattering operator $S_{\beta\alpha}$ is defined as
$$S_{\beta\alpha}:=(W^{+}_{\beta})^*W^{-}_{\alpha}.$$

We also define the Fourier transform $F_{\alpha}$ as
$$F_{\alpha}:L^2(X_a)\to L^2((E_{\alpha},\infty);L^2(C_a)),$$
by
\begin{equation}\label{myeq1.2.5}
F_{\alpha} u(\lambda,\omega):=(2\pi)^{-n_a/2}2^{-1/2}\lambda_{\alpha}^{(n_a-2)/2}\int e^{-i\lambda_{\alpha}^{1/2}\omega\cdot x_a}u(x_a)dx_a,\ \omega\in C_a
\end{equation}
where $\lambda_{\alpha}:=\lambda-E_{\alpha}$.
$\hat S_{\beta\alpha}:=F_{\beta}S_{\beta\alpha}F^*_{\alpha}$ is decomposable, namely for  $a.e.\, \lambda>\max \{E_{\alpha},E_{\beta} \}$ there exist bounded operators $\hat S_{\beta\alpha}(\lambda)\in \mathcal L(L^2(C_a);L^2(C_b))$ such that
$$(\hat S_{\beta\alpha}f)(\lambda,\hat x_b)=(\hat S_{\beta\alpha}(\lambda)f(\lambda))(\hat x_b),\ \hat x_b:=x_b/\lvert x_b\rvert\in C_b,$$
for any $f\in L^2((E_{\alpha},\infty);L^2(C_a))$ (see e.g. \cite{RS}).

The other definition comes from the asymptotic behaviors of generalized eigenfunctions even in the $N$-body problems. 

Let $\alpha$ be a non-threshold channel. When $\mu>1$ and $V_a^s\equiv 0$ for any $a\in \mathcal A$, Vasy \cite{Va} proved the following.
For $\lambda\in (E_{\alpha},\infty)\setminus \mathcal T(H)$ and $g\in C_0^{\infty}(C_a')$, for some $s>1/2$ there exists a unique generalized eigenfunction $u\in H^{\infty,-s}(X)$ of $H$, and $u$ has the form
\begin{equation}\label{myeq1.3}
u=u_{\alpha}(x^a)v_-(x_a)e^{-i\sqrt{\lambda_{\alpha}}r_a}r_a^{-(n_a-1)/2}+R(\lambda+i0)f,
\end{equation}
where $r_a:=\lvert x_a\rvert$, $v_-\in C^{\infty}(X_a)$, $\lim_{r_a\to\infty}v_-(r_a,\hat x_a)=g(\hat x_a)$, $R(z):=(H-z)^{-1}$ and $H^{\infty,-s}(X)$ is a weighted Sobolev space.
Vasy \cite{Va} defined the Poisson operator in order that the following holds: $P_{\alpha,+}(\lambda)g=u$.

Let $\beta$ be a non-threshold channel and set
$$(\pi_{\beta}u)(x_b):=\int_{X^b}u(x^b,x_b)\bar u_{\beta}(x^b)dx^b.$$
Vasy \cite{Va}  also proved that for a generalized eigenfunction $u\in H^{\infty,-s}(X),\ s>1/2$ of $H$ and for $\lambda>E_{\beta}$, $\pi_{\beta}u$ has the following distributional asymptotic behavior:
\begin{align*}
\int_{C_b}&(\pi_{\beta}u)(r_b\omega_b)h(\omega_b)d\omega_{b}\\
&=r_b^{-(n_b-1)/2}(e^{-i\sqrt{\lambda_{\beta}}r_b}Q^0_{\beta,-}(\pi_{\beta}u,h)+e^{i\sqrt{\lambda_{\beta}}r_b}Q^0_{\beta,+}(\pi_{\beta}u,h)+r_b^{-\delta}v),
\end{align*}
for $h\in C_0^{\infty}(C_b')$, where $\delta>0,\ v\in C^{\infty}((1,\infty))\cap L^{\infty}((1,\infty))$ and $Q^0_{\beta,\pm}(\pi_{\beta}u,\cdot )$ define a distribution on $C^{\infty}(C_b')$.

Vasy \cite {Va} defined the scattering matrix $\Sigma_{\beta\alpha}(\lambda)$ as
$$\Sigma_{\beta\alpha}(\lambda)(g)(h):=Q^0_{\beta,+}(\pi_{\beta}(P_{\alpha,+}(\lambda)g),h).$$
Vasy \cite{Va} also proved the following:
if $\alpha$ and $\beta$ are no-threshold channels, then for $\lambda>\max\{E_{\alpha},E_{\beta}\},\ \lambda\notin \mathcal T(H)$,
\begin{equation}\label{myeq1.4}
\hat S_{\beta\alpha}(\lambda)=C_{\beta\alpha}(\lambda)\Sigma_{\beta\alpha}(\lambda)\mathcal R,\ C_{\beta\alpha}(\lambda):=e^{i\pi(n_a+n_b-2)/4}(\lambda_{\alpha}/\lambda_{\beta})^{1/2},
\end{equation}
where $\mathcal Rg(\hat x_a):=g(-\hat x_a)$ for $g\in C_0^{\infty}(C_a')$.

Isozaki \cite{Is5} and Hassell \cite{Ha} proved similar results for 2-cluster to 3-cluster scattering matrices in three-body problems and for the free channel scattering matrices respectively using different methods.

In this paper we generalize \eqref{myeq1.4} for long-range potentials. More precisely, we obtain a result similar to \eqref{myeq1.4} for $1\geq \mu>1/2$ where, $\mu$ is as in \eqref{myeq1.2.1.1}.  In the long-range case we use the solutions to Hamilton-Jacobi equations in the definition of wave operators, and use the solutions to eikonal equations in the asymptotic behaviors of generalized eigenfunctions of $H$. These solutions are related by the Legendre transform.

Our definition of the Poisson operator $P_{\alpha,+}(\lambda)$ is similar to the one in Vasy \cite{Va}. However, we obtain the asymptotic behaviors of $u:=P_{\alpha,+}(\lambda)g$, $g\in C_0^{\infty}(C_a')$ in a way different from Vasy \cite{Va}. The reason for that is as follows. For short-range potentials the asymptotic behavior of the part corresponding to a non-threshold channel $\alpha$ of a generalized eigenfunction $u$ is expected to be as
$$u_{\alpha}(x^a)\Psi_{\alpha}(\hat x_a)e^{\pm i\sqrt{\lambda_{\alpha}}r_a}r_a^{(n_a-1)/2},$$
where $\Psi_{\alpha}\in L^2(C_a)$. However, for long-range potentials the factor $e^{\pm i \sqrt{\lambda_{\alpha}}r_a}$ in the asymptotic behavior as above is replaced by $e^{\pm iK_a(x_a,\lambda_{\alpha})}$, where $K_a(x_a,\lambda_{\alpha})$ is a solution to an eikonal equation. Since $K_a(x_a,\lambda_{\alpha})$ depends not only on $r_a$ but also on $\hat x_a$ unlike $\sqrt{\lambda_{\alpha}}r_a$, we can not reduce the study of the asymptotic behavior of $u$ to an ordinary differential equation for $r_a$ as in Vasy \cite{Va}. Instead, we define a distribution $Q^{\pm}_{\alpha}(u)$ on $C_0^{\infty}(C_a')$ as follows:
\begin{equation}\label{myeq1.5}
\begin{split}
(&Q^{\pm}_{\alpha}(u))(h)\\
&:=\lim_{\rho\to\infty}\rho^{-1}\int_{\lvert x_a\rvert<\rho}e^{\mp iK_a(x_a,\lambda_{\alpha})}r_a^{-(n_a-1)/2}h(\hat x_a)(\pi_{\alpha}u)(x_a)dx_a,
\end{split}
\end{equation}
for any $h\in C_0^{\infty}(C_a')$. The existence of the limit in \eqref{myeq1.5} is not obvious. We prove the existence of the limit in \eqref{myeq1.5} and show that $Q^{\pm}_{\alpha}(u)\in L^2(C_a)$.

Using $Q^{+}_{\beta}$ we define the scattering matrix as
$$\Sigma_{\beta\alpha}(\lambda):=Q_{\beta}^+(P_{\alpha,+}(\lambda)g).$$ Then, we obtain the relation \eqref{myeq1.4} in which the definition of $\hat S_{\beta\alpha}(\lambda)$ is also modified.

We also give a definition of the generalized Fourier transforms, and show that they are equivalent to the one defined by the wave operator approach. We also prove that the adjoint operators of generalized Fourier transforms are given by the Poisson operators $P_{\alpha,\pm}(\lambda)$.

The content of this paper is as follows. In section \ref{first.2sec} we give some preliminaries and the main results. In section \ref{secondsec} we introduce the generalized Fourier transforms for decaying potentials.  In section \ref{thirdsec} we introduce the wave operators and the scattering matrices for decaying potentials, and give the relation between the scattering matrices and the adjoint operators of the generalized Fourier transforms. In section \ref{fourthsec} we define the Poisson operators for $N$-body Schr\"odinger operators. In section \ref{fifthsec} we study the asymptotic behaviors of the generalized eigenfunctions and the solutions to nonhomogeneous equations for decaying potentials. In section \ref{fifthsec.1} we introduce the outgoing and incoming  properties and the uniqueness theorem for nonhomgeneous equations. In section \ref{fifth.1sec} we define the scattering matrices and the generalized Fourier transforms for $N$-body Schr\"odinger operators using the asymptotic behaviors of the generalized eigenfunctions, and outgoing or incoming solutions to nohomogeneous equations respectively. In section \ref{fifth.2sec} we study the asymptotic behaviors of the functions in the range of the resolvent and the Poisson operators. In section \ref{sixthsec} we prove the equivalence of the two definitions of the scattering matrices. In section \ref{eighthsec} we prove the equivalence of the two definitions of the generalized Fourier transforms and show that the adjoint operators of the generalized Fourier transforms are given by the Poisson operators.

\section{Some preliminaries and the main results}\label{first.2sec}
In this section we use the notations in section \ref{firstsec}.
We define
\begin{align}
&X_a^{\epsilon}:=\{x\in X: \mathrm{dist}(x,X_a)<\epsilon\},\notag\\
&Z_a:=X_a\setminus\bigcup_{b\nleq a}X_b,\label{myeqfirst.2.0.1}\\
&Y_a:=X\setminus\bigcup_{b\nleq a}X_b,\label{myeqfirst.2.0.1.0}\\
&Z_a^{\epsilon}:=X_a\setminus\bigcup_{b\nleq a}X_b^{\epsilon},\label{myeqfirst.2.0.1.1}\\
&Y_a^{\epsilon}:=X\setminus\bigcup_{b\nleq a}X_b^{\epsilon}.\label{myeqfirst.2.0.2}
\end{align}

The directions in which the clusters of $a$ collide are removed in $Z_a$.

We assume the potentials $V_a$ obey the following.
\begin{as}
There exists $1\geq\mu>1/2$ such that  for any $a\in \mathcal A$, $V_a(x^a)=V_a^s(x^a)+V_a^l(x^a)$, where
\begin{itemize}
\item[(1)]$V_a^s(x^a)$ is compactly supported and $V_a^s(-\Delta^a+1)^{-1}$ is compact.
\item[(2)]$V_a^l(x^a)\in C^{\infty}(X^a)$ and for any $\gamma\in \mathbb N^{n_a}$
$$\lvert\partial^{\gamma} V_a(x^a)\rvert=\mathcal O(\lvert x^a\rvert^{-\mu-\lvert \gamma\rvert}).$$
\end{itemize}
\end{as}

Let $H$ and $H^a$ be as in \eqref{myeq1.2.1} and \eqref{myeq1.2.2} respectively.

Set $I_a(x):=\sum_{b\nleq a}V_b(x^b)$, $I_a^l(x):=\sum_{b\nleq a}V^l_b(x^b)$, and let $\chi_0\in C^{\infty}(\mathbb R)$ be a function which is $0$ for $x<1$ and $1$ for $x>2$. We define the modified potentials as follows:
 $$\tilde I_a(x_a):=I_a^l(x_a)\prod_{a< b}\chi_0\left(\frac{\langle x_a^b\rangle}{\langle x_a\rangle}\log\langle x_a\rangle\right)\chi_0\left(\frac{\lvert x_a\rvert}{K}\right),$$
where $K>0$ is a constant and $x^b_a$ is defined by \eqref{myeq1.2.0.1}. Let $\mu_0$ be a constant such that $\mu_0\in(0,\mu)$. Then if we choose $K$ large enough, we have
 $$\partial_{x_a}^{\gamma}\tilde I_a(x_a)=\mathcal O\left(\lvert x_a\rvert^{-\mu_0-\lvert\gamma\rvert}\right),$$
as $\lvert x_a\rvert\to\infty$ for any $\gamma\in \mathbb N^{n_a}$.
 
In section \ref{thirdsec} we show that there exists $S_{a,\pm}(\xi_a,t)\in C^{\infty}(X_a\setminus \{0\}\times \mathbb R_{\pm})$ satisfying the following : For any compact set $A\in X_a\setminus \{0\}$ there exists $T>0$ such that
$$\frac{\partial S_{a,\pm}}{\partial t}(\xi_a,t)=\lvert \xi_a\rvert^2+\tilde I_a(\nabla_{\xi_a}S_{a,\pm}(\xi_a,t)).$$
 
The time-dependent definition of the scattering matrices is based on the wave operators. If $\alpha$ is a non-threshold channel, then there exist wave operators
$$W_{\alpha}^{\pm}:=s-\lim_{t\to\pm\infty}e^{itH}J_{\alpha}e^{-i(S_{a,\pm}(D_a,t)+E_{\alpha}t)},$$
where $e^{-i(S_{a,\pm}(D_a,t)+E_{\alpha}t)}f:=F_{\alpha}^*(e^{-i(S_{a,\pm}(\xi_a,t)+E_{\alpha}t)}(F_{\alpha}f)(\lvert\xi_a\rvert^2,\hat\xi_a))$, $\hat \xi_a:=\xi_a/\lvert \xi_a\rvert$.

The scattering operator is defined as
$$S_{\beta\alpha}:=(W_{\beta}^+)^*W_{\alpha}^-.$$

Then $\hat S_{\beta\alpha}:=F_{\beta}S_{\beta\alpha}F^*_{\alpha}$ is decomposable, namely for $a.e.\ \lambda>\max \{E_{\alpha},E_{\beta} \}$ there exist bounded operators $\hat S_{\beta\alpha}(\lambda)\in \mathcal L(L^2(C_a);L^2(C_b))$ such that
$$(\hat S_{\beta\alpha}f)(\lambda,\hat x_b)=(\hat S_{\beta\alpha}(\lambda)f(\lambda))(\hat x_b).$$

The stationary definition of the scattering matrices comes from the asymptotic behaviors of generalized eigenfunctions. Let  $L^{2,l}(\mathbb R^n)$ denote the Hilbert space of all measurable functions on $\mathbb R^n$ such that
$$\lVert f\rVert_{l}^2=\int(1+\lvert x\rvert)^{2l}\lvert f(x)\rvert^2dx<\infty.$$
We also need the spaces $\mathcal B(\mathbb R^n)$ and $\mathcal B^*(\mathbb R^n)$ of functions. We set
$$\rho_j,\ :=2^j,\ (j\in \mathbb N,\ j\geq 0)$$
$$\Omega_0:=\{x\in \mathbb R^n: \lvert x\rvert<1\},\ \Omega_j:=\{x\in \mathbb R^n: 2^{j-1}<\lvert x\rvert<2^j\},\ (j\in \mathbb N,\ j\geq 1).$$
Let $\mathcal B(\mathbb R^n)$ be the set of functions $u$ such that
$$\Vert u\rVert_{\mathcal B(\mathbb R^n)}:=\sum_{j=0}^{\infty}\rho_j^{1/2}\lVert u\rVert_{L^2(\Omega_j)}<\infty.$$
Then the dual space $\mathcal B^*(\mathbb R^n)$ of $\mathcal B(\mathbb R^n)$ is the set of functions $u$ such that
$$\lVert u\rVert_{\mathcal B^*(\mathbb R^n)}:=\sup_{j\geq 0}\rho_j^{-1/2}\lVert u\rVert_{L^2(\Omega_j)}<\infty.$$
Moreover, there exists a constant $C>0$ such that
$$C^{-1}\lVert u\rVert_{\mathcal B^*(\mathbb R^n)}\leq \left(\sup_{\rho>1}\rho^{-1}\int_{\lvert x\rvert<\rho}\lvert u(x)\rvert^2dx\right)^{1/2}\leq C\lVert u\rVert_{\mathcal B^*(\mathbb R^n)}.$$

The relation between $L^{2,l}(\mathbb R^n)$,  $\mathcal B(\mathbb R^n)$ and $\mathcal B^*(\mathbb R^n)$ is as follows: when $l>1/2$
\begin{align*}
L^{2,l}(\mathbb R^n)\subset \mathcal B(\mathbb R^n) \subset L^{2,1/2}(\mathbb R^n)&\subset L^{2}(\mathbb R^n)\\
&\subset L^{2,-1/2}(\mathbb R^n)\subset \mathcal B*(\mathbb R^n)\subset L^{2,-l}(\mathbb R^n).
\end{align*}

Since we can assume $X_a=\mathbb R^n$ for some $n\in \mathbb N$, we can define $L^{2,l}(X_a)$, $\mathcal B(X_a)$ and $\mathcal B^*(X_a)$ in the same way as above.

Let $\alpha$ be a non-threshold channel. We show in section \ref{fourthsec} that for $\lambda\in (E_{\alpha},\infty)\setminus \mathcal T(H)$ and $g\in C_0^{\infty}(C_a')$ there exists a unique function $u\in \mathcal B^*(X)$ such that $(H-\lambda)u=0$, and $u-e^{-iK_{a}(x_a,\lambda_{\alpha})}r_a^{-(n_a-1)/2}u_{\alpha}(x^a)g(\hat x_a)$ is outgoing, where $K_a(x_a,\lambda_{\alpha})$ is the solution to the eikonal equation $\lvert \nabla_aK_a(x_a,\lambda_{\alpha})\rvert^2+\tilde I_a(x_a)=\lambda_{\alpha}$ in section \ref{secondsec} with $V$ replaced by $\tilde I_a$ (for the definition of outgoing and incoming properties see section \ref{fifthsec.1}). We define the Poisson operator $P_{\alpha,+}(\lambda) : C_0^{\infty}(C_a')\to \mathcal B^*(X)$ by $P_{\alpha,+}(\lambda)g:=u$. There is also  a unique function $u\in \mathcal B^*(X)$ such that $(H-\lambda)u=0$, and $u-e^{iK_{a}(x_a,\lambda_{\alpha})}r_a^{-(n_a-1)/2}u_{\alpha}(x^a)g(\hat x_a)$ is incoming, and we define $P_{\alpha,-}(\lambda)$ by $P_{\alpha,-}(\lambda)g:=u$.

Let $\beta$ be a non-threshold channel. In section \ref{fifth.1sec} we show that for a generalized eigenfunction $u\in \mathcal B^*(X)$ and $h\in C_0^{\infty}(C_b')$ the limit
\begin{equation}\label{myeqfirst.2.1}
Q_{\beta}^{\pm}(u)(h):=\lim_{\rho\to\infty}\rho^{-1}\int_{\lvert x_b\rvert<\rho}h(\hat x_b)e^{\mp iK_b(x_b,\lambda_{\beta})}r_b^{-(n_b-1)/2}(\pi_{\beta}u)(x_b)dx_b,
\end{equation}
exists.  We can see that $Q_{\beta}^{\pm}(u)$ defines a distribution on $C_b'$ and actually $Q_{\beta}^{\pm}(u)\in L^2(C_b)$. We define the scattering matrix $\Sigma_{\beta\alpha}(\lambda) : C_0^{\infty}(C_a')\to L^2(C_b)$ by
$$\Sigma_{\beta\alpha}(\lambda)g:=Q_{\beta}^+(P_{\alpha,+}(\lambda)g).$$

One of our main result is the following
\begin{thm}\label{sme}
Let $\alpha$ and $\beta$ be non-threshold channels. Then, 
$$\hat S_{\beta\alpha}(\lambda)g=e^{i\pi(n_a+n_b-2)/4}\lambda_{\alpha}^{1/4}\lambda_{\beta}^{-1/4}\Sigma_{\beta\alpha}(\lambda)\mathcal Rg,$$
for $g\in C_0^{\infty}(C_a')$ and $a.e.\, \lambda>\max\{E_{\alpha}, E_{\beta}\}$ where $(\mathcal R g)(\hat x_a):=g(-\hat x_a)$.
\end{thm}

We also give the stationary definition of the generalized Fourier transforms. For $\lambda\in\sigma_{ess}(H)\setminus(\sigma_{pp}(H)\cup \mathcal T(H))$ the resolvent $R(\lambda\pm i\epsilon):=(H-\lambda\mp i\epsilon)^{-1}$, $\epsilon>0$ is extended to $R(\lambda\pm i0)$ as an operator in $\mathcal L(\mathcal B(X),\mathcal B^*(X))$, where $\sigma_{ess}(A)$ is the set of essential spectra of $A$ (see \cite {JP}). Let $f\in L^{2,l}(X),\ l>1/2$, $\alpha$ be a non-threshold channel, and set $u=R(\lambda+i0)f$ or $u=R(\lambda-i0)f$. Then, the limit $Q_{\alpha}^{\pm} (u)(h)$ as \eqref{myeqfirst.2.1} exits. We define the generalized Fourier transforms $\mathcal G_{\alpha}^{\pm}(\lambda) : L^{2,l}(X)\to L^2(C_a)$ by
$$\mathcal G_{\alpha}^+(\lambda)f=D_{\alpha}^+(\lambda)Q_{\alpha}^+(R(\lambda+i0)f),$$
and
$$\mathcal G_{\alpha}^-(\lambda)f=D_{\alpha}^-(\lambda)\mathcal RQ_{\alpha}^-(R(\lambda-i0)f),$$
where $D_{\alpha}^{\pm}(\lambda):=e^{\pm i\pi(n_a-3)/4}\pi^{-1/2}\lambda_{\alpha}^{1/4}$.

The generalized Fourier transforms are related to the wave operators and Poisson operators as in the following theorem.
\begin{thm}\label{gfe}
Let $\alpha$ be a non-threshold channel. Then, for $f\in L^{2,l}(X),\ l>1/2$ and $h\in C_0^{\infty}(C_a')$ we have
\begin{itemize}
\item[(1)]
$$(F_{\alpha}(W_{\alpha}^+)^*f)(\lambda)=\mathcal G_{\alpha}^{+}(\lambda)f,$$
$$(F_{\alpha}(W_{\alpha}^-)^*f)(\lambda)=\mathcal G_{\alpha}^{-}(\lambda)f,$$
for $a.e.\, \lambda>E_{\alpha}$.
\item[(2)]
$$(\mathcal G_{\alpha}^+(\lambda))^*=-\tilde D_{\alpha}^-(\lambda)P_{\alpha,-}(\lambda),$$
$$(\mathcal G_{\alpha}^-(\lambda))^*=-\tilde D_{\alpha}^+(\lambda)P_{\alpha,+}(\lambda)\mathcal R,$$
for $a.e.\, \lambda>E_{\alpha}$, where $\tilde D_{\alpha}^{\pm}(\lambda):=2^{-1}i\pi^{-1/2}\lambda_{\alpha}^{-1/4}e^{\pm i\pi(n_a-3)/4}$.
\end{itemize}
\end{thm}

\section{Generalized Fourier transforms for decaying potentials}\label{secondsec}
In this section we suppose $n\in \mathbb N$, and that the potential $V(x)$ is a real valued function such that $V\in C^{\infty}(\mathbb R^n)$, and for some $\mu>0$,
\begin{equation}\label{myeq2.0.1}
\lvert\partial^{\gamma} V(x)\rvert=\mathcal O(\lvert x\rvert^{-\mu-\lvert \gamma\rvert})\ \mathrm{as}\ \lvert x\rvert\to +\infty,
\end{equation}
for any multi-index $\gamma$.

The oscillations in the asymptotic behaviors of the functions $\tilde R(\lambda\pm i0)f,\ f\in \mathcal B(\mathbb R^n)$ are given by the solutions to the eikonal equations, where $\tilde R(z):=(\tilde H-z)^{-1}$, $\tilde H:=-\Delta+V$.  The solutions to the eikonal equations are given by the following lemma.

\begin{lem}[{\cite[Lemma 2.1]{II}}, {\cite[Theorem4.1]{Is}}]\label{Ydef}
There exists a real valued function $Y(x,\lambda)\in C^{\infty}(\mathbb R^n\times \mathbb R_+),\ \mathbb R_+:=(0,\infty)$ satisfying the following properties:
\begin{itemize}
\item[(1)]For any compact set $\Lambda\in \mathbb R_+$, there exists a constant $R_0(\Lambda)$ such that
\begin{equation}\label{myeq2.1}
2\sqrt{\lambda}\frac{\partial Y}{\partial r}(x,\lambda)=\lvert \nabla_xY(x,\lambda)\rvert^2+V(x),
\end{equation}
for $\lvert x\rvert >R_0(\Lambda),\ \lambda\in \Lambda$.
\item[(2)] For any $\alpha\in \mathbb N^n,\ k\in \mathbb N$ and any compact set $\Lambda \in \mathbb R_+$ we have
$$\left \lvert \left(\frac{\partial}{\partial x}\right)^{\alpha}\left(\frac{\partial}{\partial \lambda}\right)^kY(x,\lambda)\right\rvert\leq C_{\alpha,k,\Lambda}\langle x\rangle^{1-\lvert\alpha\rvert-\mu},\ \lambda\in \Lambda.$$
\end{itemize}
\end{lem}

If we put 
\begin{equation*}
K(x,\lambda)=\sqrt{\lambda}r-Y(x,\lambda),
\end{equation*}\label{myeq2.1.1}
where $r:=\lvert x\rvert$, then by \eqref{myeq2.1} we can see $K(x,\lambda)$ satisfies the eikonal equation
$$\lvert \nabla_xK(x,\lambda)\rvert^2+V(x)=\lambda.$$
Set 
\begin{equation}\label{myeq2.2}
w_{\pm}:=r^{-(n-1)/2}e^{\pm i(K(x,\lambda)-\pi(n-3)/4)}.
\end{equation}
Then, the asymptotic behavior of $\tilde R(\lambda\pm i0)f,\ f\in \mathcal B(\mathbb R^n)$ is given by the following lemma.
\begin{lem}[{\cite[Theorem 3.5]{GY}}, \cite{Is}]
For any $f\in \mathcal B(\mathbb R^n)$ and $\lambda>0$, there exist $a_{\pm}\in L^2(\mathbb S^{n-1})$ such that, as $x\to \infty$,
$$(\tilde R(\lambda\pm i0)f)(x)=\pi^{1/2}\lambda^{-1/4}a_{\pm}(\pm\hat x)w_{\pm}(x,\lambda)+o_{av}(\lvert x\rvert^{-(n-1)/2}),$$
where $\hat x:=x/\lvert x\rvert$, and $f(x)=o_{av}((\lvert x\rvert^{-(n-1)/2})$ means that
$$\lim_{\rho\to+\infty}\rho^{-1}\int_{\lvert x\rvert\leq \rho}\lvert f(x)\rvert^2dx=0.$$
\end{lem}

We define the mapping $\mathcal F_{\pm}(\lambda) : \mathcal B(\mathbb R^n) \to L^2(\mathbb S^{n-1})$ by $(\mathcal F_{\pm}(\lambda)f)(\hat x) = a_{\pm}(\hat x)$. It is known that $\mathcal F_{\pm}(\lambda)$ is bounded from $\mathcal B(\mathbb R^n)$ into $L^2(\mathbb S^{n-1})$ (see \cite{GY}). Set $\hat {\mathcal H}=L^2(\mathbb R_+; L^2(\mathbb S^{n-1}))$. For any $f\in \mathcal B(\mathbb R^n)$ we define the mapping $F_{\pm}: \mathcal B(\mathbb R^n)\to\hat{\mathcal H}$ by $(F_{\pm}f)(\lambda)=\mathcal F_{\pm}(\lambda)f$. Then, $F_{\pm}$ is uniquely extended to the partial isometry with initial set $\mathcal H_{ac}(\tilde H)$ and final set $\hat{\mathcal H}$, where $\mathcal H_{ac}(\tilde H)$ is the absolutely continuous subspace $\tilde H$ (see \cite[Theorem 5.2]{GY} and \cite[Theorem 2.5]{II}).

We have explicit representations for $\mathcal F_{\pm}^*(\lambda):=(\mathcal F_{\pm}(\lambda))^*$ on $a\in C^{\infty}(\mathbb S^{n-1})$. Let $\eta\in C^{\infty}(\mathbb R)$ be a function satisfying the following condition:
there exists $\kappa>0$ such that
\begin{equation}\label{myeq2.3}
\eta(t)=\begin{cases}
&0\qquad (t< \kappa)\\
&1\qquad (t>2\kappa).
\end{cases}
\end{equation}
Then, the functions
\begin{equation}\label{myeq2.4}
u_{\pm}(x,\lambda)=\eta(r)a(\pm \hat x)w_{\pm}(x,\lambda),
\end{equation}
belong to the space $\mathcal B^*(\mathbb R^n)$. We set $g_{\pm}(\lambda) = (-\Delta +V-\lambda)u_{\pm}(\lambda)$. Then, by straightforward calculations we can see $g_{\pm}(\lambda)\in L^{2,l}(\mathbb R^n)$ for some $l>1/2$.
\begin{lem}[{\cite[Lemma 3.7, Corollary 3.8]{GY}}, {\cite[Theorem 3.4]{II}}]\label{Frepre}
Let $a\in C^{\infty}(\mathbb S^{n-1})$, $w_{\pm}$ be defined by \eqref{myeq2.2} and $u_{\pm},\ g_{\pm}$ be as above. Then we have
$$\pm 2i\pi^{1/2}\lambda^{1/4}\mathcal F_{\pm}^*(\lambda)a=u_{\pm}(\lambda)-\tilde R(\lambda\mp i0)g_{\pm}(\lambda).$$
Moreover, for any $a\in L^2(\mathbb S^{n-1})$ we have $\mathcal F_{\pm}^*(\lambda)a\in H^2_{loc}(\mathbb R^n)\cap \mathcal B^*(\mathbb R^n)$ and
$$(-\Delta + V-\lambda)\mathcal F^*_{\pm}(\lambda)a=0.$$
\end{lem}

We introduce the following class of symbols of pseudodifferential operators.
Let $S^{m,l},\ m,l \in \mathbb R$, be the symbol class of $C^{\infty}(\mathbb R^n\times \mathbb R^n)$-functions $p(x,\xi)$ satisfying the following condition: for any $\alpha, \beta\in\mathbb N^n$ there exists $C_{\alpha\beta}>$ such that
$$\lvert \partial_{x}^{\alpha}\partial_{\xi}^{\beta}p(x,\xi)\rvert\leq C_{\alpha,\beta}\langle x\rangle^{l-\lvert \alpha\rvert}\langle \xi\rangle^m,\ \forall (x,\xi)\in \mathbb R^n\times \mathbb R^n.$$
The corresponding pseudodifferential operators are defined by the Weyl quantization:
$$(Op(p)\psi)(x)=(2\pi)^{-n}\int\int e^{i(x-y)\cdot\xi}p((x+y)/2,\xi)\psi(y)dyd\xi,$$
for $\psi \in \mathcal S(\mathbb R^n)$. We also define the right quantization of $p$ by
$$(Op^r(p)\psi)(x)=(2\pi)^{-n}\int\int e^{i(x-y)\cdot\xi}p(y,\xi)\psi(y)dyd\xi,$$
for $\psi \in \mathcal S(\mathbb R^n)$.

We need the following micro-local resolvent estimate.
\begin{lem}\label{micro1}
Let $\lambda>0$, $s>1/2$, $t>1$ and $p_{\pm}\in S^{0,0}$ be a symbol satisfying the following condition: there exists $0<\epsilon$ such that $p_{\pm}(x,\xi)=0$ if $\pm \hat x\cdot \hat\xi<1-\epsilon$, where $\hat x:=x/\lvert x\rvert$ and $\hat \xi:=\xi/\lvert \xi\rvert$. Then there exists $C>0$ such that
$$\lVert Op(p_{\mp})\tilde R(\lambda\pm i0)f\rVert_{s-t}\leq C\lVert f\rVert_s.$$
\end{lem}
This lemma is essentially due to Skibsted \cite{Sk} (see also Isozaki \cite[Theorem 2.2]{Is4}). The estimate for $\tilde R(\lambda+i0)$ is obtained by the argument similar to the one below \cite[Theorem 2.2]{Is4} for the decaying potential $V$, with $s$ in \cite[Theorem 2.2]{Is4} replaced by $s-t$. The estimate for $\tilde R(\lambda-i0)$ follows from the similar propagation estimate as $t\to-\infty$. Note that we can assume $s-t>-1/2$ since otherwise there exists $t'>1$ such that $t>t'$, $s-t'>-1/2$ and $\langle x\rangle^{s-t}=\langle x\rangle^{-(t-t')}\langle x\rangle^{s-t'}$.

We also have the micro-local estimate for $u_{\pm}(\lambda)$.
\begin{lem}\label{micro2}
Let $p_{\mp}(x,\xi)$ be a symbol such that $p_{\mp}(x,\xi)=0$ for $(x,\xi)$ satisfying one of the following conditions for some $\epsilon>0$
\begin{itemize}
\item[(i)]$\lvert\lvert \xi\rvert^2-\lambda\rvert<\epsilon$,
\item[(ii)]$\pm\hat x\cdot\hat \xi>1-\epsilon$ for some $\epsilon>0$,
\item[(iii)]$\hat\xi \in \mathrm{supp}\, a(\pm\cdot)$.
\end{itemize}
Then for any $m\in \mathbb N$ there exists a constant $C$ such that
$$\lvert\vert x\rvert^mOp(p_{\mp})u_{\pm}\rvert\leq C,$$
where $u_{\pm}$ is defined in \eqref{myeq2.4}.
\end{lem}
\begin{proof}
Since a support of a symbol does not change by a choice of quantization except an error in $S^{-\infty,-\infty}$, we can consider the right quantization $Op^r(p_{\mp})$ instead of the Weyl quantization.
We can write
\begin{align*}
&Op^r(p_{\mp})u_{\pm}\\
&\quad=\lim_{\mu\to0_+}\int\chi_0(\mu\xi)\chi_0(\mu y) e^{i(x-y)\cdot\xi\pm iK(y)}p_{\mp}(y,\xi)a(\pm \hat y)\eta(y)\lvert y\rvert^{-(n-1)/2}dyd\xi.
\end{align*}
Since there exists a constant $\tilde C>0$ such that
$$\lvert \xi\mp\sqrt{\lambda}\hat y\rvert>\tilde C^{-1}(\lvert \xi\rvert+1),$$
for $(y,\xi)\in \mathrm{supp}\, p_{\mp}(y,\xi)a(\pm \hat y)$ and we have
$$i\lvert \xi\mp\sqrt{\lambda}\hat y\rvert^{-2}(\xi\mp\sqrt{\lambda}\hat y)\cdot\nabla_ye^{-iy\cdot\xi\pm i\sqrt{\lambda}\lvert y\rvert}=e^{-iy\cdot\xi\pm i\sqrt{\lambda}\lvert y\rvert},$$
noting that $\lvert\partial_y^{\alpha}(a(\pm \hat y))\rvert=\mathcal O(\lvert y\rvert^{-\lvert \alpha\rvert})$ as $\lvert y\rvert\to\infty$, we obtain the result by integration by parts.
\end{proof}

\section{Wave operators and scattering matrices for decaying potentials}\label{thirdsec}
In this section we suppose $n\in \mathbb N$, and \eqref{myeq2.0.1} for the potential $V$ and we use the notations in section \ref{secondsec}. To define the wave operator we need the solutions $S_{\pm}(\xi,t)$ to the Hamilton-Jacobi equations obtained by the Legendre transformation of $K(x,\lambda)$. As in \cite[Lemma 6.1]{II} we have
\begin{lem}[\cite{II} Lemma 6.1]
There exist $x_{\pm}(\xi,t),\ \lambda_{\pm}(\xi,t)\in C^{\infty}((\mathbb R^n\setminus\{0\}\times \mathbb R_{\pm})$ satisfying the following  condition : For any compact set $\Lambda\subset \mathbb R^n\setminus\{0\}$ there exist positive constants $T,C$ such that for $\xi\in \Lambda$ and $\pm t>T$ we have
$$\xi=\pm \frac{\partial K}{\partial x}(x_{\pm}(\xi,t), \lambda_{\pm}(\xi,t)),\ t=\pm\frac{\partial K}{\partial \lambda}(x_{\pm}(\xi,t),\lambda_{\pm}(\xi,t)),$$
$$\lvert x_{\pm}(\xi,t)-2\xi t\rvert\leq C(1+\lvert t\rvert)^{1-\mu},\ \lvert \lambda_{\pm}(\xi,t)-\lvert \xi\rvert^2\rvert\leq C(1+\lvert t\rvert)^{-\mu}.$$
\end{lem}
\begin{rem}
Although only $x_{+}(\xi,t)$ and $\lambda_{+}(\xi,t)$ are considered in \cite{II}, the result for $x_-(\xi,t)$ and $\lambda_-(\xi,t)$ is obtained in the same way.
\end{rem}

\begin{lem}[\cite{II}]\label{HamiJac}
Let us define
$$S_{\pm}(\xi,t)=x_{\pm}(\xi,t)\xi+\lambda_{\pm}(\xi,t)t\mp K(x_{\pm}(\xi,t),\lambda_{\pm}(\xi,t)).$$
Then, for any compact set $\Lambda\subset \mathbb R^n\setminus\{0\}$ there exists $T>0$ such that $\nabla_{\xi} S_{\pm}(\xi,t)=x_{\pm}(\xi,t)$, $(\partial S_{\pm}/\partial t)(\xi,t)=\lambda_{\pm}(\xi,t)$ and
$$\frac{\partial S_{\pm}}{\partial t}(\xi,t)= \lvert \xi\rvert^2+V(\nabla_{\xi}S_{\pm}(\xi,t)),$$
for $\xi\in \Lambda$ and $\pm t>T$.
\end{lem}

Let $\mathbf{F}_{\pm} : L^2(\mathbb R^n)\to L^2(\mathbb R^n)$ be defined by
$$(\mathbf F_{\pm}f)(\xi)=2^{1/2}\lvert \xi\rvert^{-(n-2)/2}(F_{\pm}f)(\lvert \xi\rvert^2,\xi/\lvert\xi\rvert).$$
Then $\mathbf F_{\pm}$ is a partial isometry on $L^2(\mathbb R^n)$ with the initial set $\mathcal H_{ac}(\tilde H)$ and the final set $L^2(\mathbb R^n)$.

Let $\mathbf F_0$ be the ordinary Fourier transformation:
$$(\mathbf F_0f)(\xi):=(2\pi)^{-n/2}\int e^{-ix\cdot\xi}f(x)dx.$$
Then, we have the following lemma.
\begin{lem}[{\cite[Theorem 7.3]{II}}]\label{wave}
The wave operators
$$W_{\pm}=s-\lim_{t\to\pm\infty}e^{it\tilde H}e^{-iS_{\pm}(D,t)},$$
exist and we have $W_{\pm}=\mathbf F_{\pm}^*\mathbf F_0$. Here $e^{-iS_{\pm}(D,t)f}:=\mathbf F_0^{*}[e^{-iS_{\pm}(\xi,t)}\mathbf F_0f]$. Moreover, we have the intertwining property: for any bounded Borel function $\varphi$ on $\mathbb R$ we have $\varphi(\tilde H)W_{\pm}=W_{\pm}\varphi(-\Delta)$.
\end{lem}

\begin{rem}
This result was proved in \cite{II} for $S_+(\xi,t)$, but the one for $S_-(\xi,t)$ is obtained in the same way. The intertwining property follows from $\mathbf F_{\pm}\varphi(\tilde H)=\varphi(\lvert \xi\rvert^2)\mathbf F_{\pm}$ and $\varphi(\lvert\xi\rvert^2)\mathbf F_0=\mathbf F_0\varphi(-\Delta)$, where $\lvert \xi\rvert^2$ is the multiplication operator by $\lvert \xi\rvert^2$.
\end{rem}

The scattering operator $S$ is defined by $S:=W_+^*W_-$. By the intertwining property $\hat S:=\mathbf F_0S\mathbf F_0^{*}$ is decomposable (see \cite{RS}). Denoting the fibers of $\hat S$ by $\hat S(\lambda)$, Lemma \ref{wave} implies $\hat S(\lambda):\mathcal F_-(\lambda)f\mapsto \mathcal F_+(\lambda)f.$ for any $f\in \mathcal B(\mathbb R^n)$.

We have the relation between the asymptotic behaviors of the generalized eigenfunctions and the scattering matrices.
\begin{lem}[\cite{GY}]\label{2bodyasympt}
Let $w_{\pm}(x,\lambda)$ be defined by \eqref{myeq2.2}. Then for any $a\in C^{\infty}(\mathbb S^{n-1})$ we have
\begin{align*}
&(\mathcal F_+^*(\lambda)a)(x)\\
&\quad=C(\lambda)\left(a(\hat x)w_+(x,\lambda)-(\mathcal R\hat S^*(\lambda)a)(\hat x)w_-(x,\lambda)\right)+o_{av}(\lvert x\rvert^{-(n-1)/2}),
\end{align*}
and
\begin{align*}
&(\mathcal F_-^*(\lambda)a)(x)\\
&\quad=C(\lambda)\left((\hat S(\lambda)a(\hat x))w_+(x,\lambda)-(\mathcal Ra)(\hat x)w_-(x,\lambda)\right)+o_{av}(\lvert x\rvert^{-(n-1)/2}),
\end{align*}
where $C(\lambda):=-i2^{-1}\pi^{-1/2}\lambda^{-1/4}$ and $(\mathcal Ra)(\hat x):=a(-\hat x)$.
\end{lem}

\section{Poisson operators}\label{fourthsec}
In this section we use the notations in sections \ref{firstsec} and \ref{first.2sec}. Let $\alpha$ be a non-threshold channel.

Set $g\in C^{\infty}_0(C_a')$ and
\begin{equation}\label{myeq4.1}
v_{\alpha}^{\pm}(\lambda,x_a):=\eta(r_a)g(\hat x_a) r_a^{-(n_a-1)/2}e^{\pm iK_{a}(x_a,\lambda_{\alpha})}.
\end{equation}
Here, $\eta$ and $K_a(x_a,\lambda_{\alpha})$ are defined by \eqref{myeq2.3} for some $\kappa>0$ and
$$K_a(x_a,\lambda_{\alpha}):=\sqrt{\lambda_{\alpha}}r_a-Y_a(x_a,\lambda_{\alpha}),$$
respectively, where $\lambda_{\alpha}:=\lambda-E_{\alpha}$, $\lambda>E_{\alpha}$ and $Y_a(x_a,\lambda_{\alpha})$ is the function obtained in Lemma \ref{Ydef} with $\lambda$ and $V(x)$ replaced by $\lambda_{\alpha}$ and $\tilde I_a(x_a)$ respectively.

We can easily see that $(H-\lambda)(J_{\alpha}v_{\alpha}^{\pm})\in L^{2,l}(X),\ l>1/2$. Thus, we can define the Poisson operator $P_{\alpha,\pm}(\lambda): C_0^{\infty}(C_a') \to \mathcal B^*(X)$ by
$$P_{\alpha,\pm}(\lambda)\varphi =J_{\alpha}v_{\alpha}^{\mp}-R(\lambda\pm i0)(H-\lambda)(J_{\alpha}v_{\alpha}^{\mp}).$$

\section{Asymptotic behaviors of generalized eigenfunctions and solutions to nonhomogeneous equations for decaying potentials}\label{fifthsec}
In this section we suppose $n\in \mathbb N$, $V\in C^{\infty}(\mathbb R^n)$ satisfy \eqref{myeq2.0.1}, and use the notations in  section \ref{secondsec}.
\begin{pro}\label{asymptotic}
Let $\lambda>0$. Suppose $u\in \mathcal B^*(\mathbb R^n)$ satisfy $(\tilde H-\lambda)u\in L^{2,l}(U),\ l>1/2$. Here $U$ is a conic region, namely there exists $C>0$ such that for any $c>1$ and $x\in U$, $\lvert x \rvert>C$ we have $cx\in U$, and $L^{2,l}(U):=\{f:\lVert f\rVert_l^2:=\int_U\langle x\rangle^{2l}\lvert f(x)\rvert^2dx<\infty\}$. Then for $h\in C_0^{\infty}(U')$ the following limit exists:
$$\lim_{\rho\to\infty}\rho^{-1}\int_{\lvert x\rvert<\rho}e^{\mp iK(x,\lambda)}r^{-(n-1)/2}h(\hat x)u(x)dx,$$
where $U':=U\cap \mathbb S^{n-1}$. The limit is equal to
\begin{align*}
2^{-1}i\lambda^{-1/2}&\int_{\mathbb R^n} \{((\tilde H-\lambda)u)(x)\bar v_{\pm}(x)\\
&-u(x)((\tilde H-\lambda)\bar v_{\pm})(x)\}dx,
\end{align*}
where $v_{\pm}(x):=\eta(r)h(\hat x)e^{\pm iK(x,\lambda)}r^{-(n-1)/2}$ with $\eta$ defined by \eqref{myeq2.3} for $\kappa$ large enough.
\end{pro}

To prove Proposition \ref{asymptotic} we need some lemmas.

Let $\psi\in C_0^{\infty}(\mathbb R)$ be a function such that $\psi=1$ near $\lambda$. Set $\psi_1(y):=1-\psi(y)$ and $\psi_2(y):=\psi_1(y)(y-\lambda)^{-1}$. Then, we can write
$$u=\psi(\tilde H)u+\psi_1(\tilde H) u.$$

\begin{lem}\label{energycut}
We have
\begin{equation}\label{myeq5.0.0.1}
\lim_{\rho\to\infty}\rho^{-1}\int_{\lvert x\rvert<\rho}e^{\mp iK(x,\lambda)}r^{-(n-1)/2}h(\hat x)\eta(r)(\psi_1(\tilde H)u)(x)dx=0,
\end{equation}
\begin{equation}\label{myeq5.0.0.2}
\lim_{\rho\to\infty}\rho^{-1}\int_{\lvert x\rvert<\rho}e^{\mp iK(x,\lambda)}r^{-(n-1)/2}h(\hat x)\eta(r)((\psi(\tilde H)-\psi(-\Delta))u)(x)dx=0.
\end{equation}
\end{lem}
\begin{proof}
Let $\phi(x)\in C^{\infty}(\mathbb R^n)$  be a function such that $\phi(x)h(\hat x)=h(\hat x)$ for $\lvert x\rvert>1$, and $\tilde \phi(x)=1-\phi(x)$. Then, $h(\hat x)\eta(r)\psi_1(\tilde H)(\tilde \phi(x) u(x))\in L^{2,l}(\mathbb R^n)$ for any $l>0$ and $h(\hat x)\eta(r)\psi_1(\tilde H)(\phi(x) u(x))=h(\hat x)\eta(r)\psi_2(\tilde H)(\tilde H-\lambda)\phi(\tilde H) u\in L^{2,l}(\mathbb R^n)$ for $l>1/2$. Thus, \eqref{myeq5.0.0.1} holds.

We can easily confirm \eqref{myeq5.0.0.2} using Hellfer-Sj\"ostrand formula:
$$\psi(A):=\frac{1}{2\pi i}\int_{\mathbb C}\bar{\partial}_z\Psi(z)(z-A)^{-1}dz\wedge d\bar z,$$
where $\Psi$ is the almost analytic extension of $\psi$ (see e.g. \cite{HS}).

\end{proof}

Set $t_0(x,\xi)=\eta(\lvert x\rvert)\psi(\lvert\xi\rvert^2)$. By Lemma \ref{energycut} we only need to prove for $u_0:=Op(t_0)u$ and $v_{\pm}$ the existence of the limit 
$$\lim_{\rho\to \infty}\rho^{-1}\int_{\lvert x\rvert<\rho}u_0(x) \bar v_{\pm}(x)dx.$$
We consider the case of $v_+$ and denote $v_+$ by $v$. The case of $v_-$ is similar.

Let $\chi_{\pm}\in C^{\infty}(\mathbb R)$ be functions such that  $\chi_{+}(s)=1$ for $s>\sqrt{\lambda}/2$, $\chi_{+}(s)=0$ for $s<-\sqrt{\lambda}/2$, and $\chi_{+}+\chi_{-}=1$.
Set $t_{\pm}:=\chi_{\pm}(\hat x\cdot \xi)t_0(x,\xi)$ and $T_{\pm}:=Op(t_{\pm})$.

Then we can decompose $u_0$ as $u_0=u_++u_-$ where $u_{\pm}:=T_{\pm}u$.

Set $\mathcal D^{\pm}:=-i\partial_r-i(n-1)/2r\mp\partial_rK$. 

\begin{lem}\label{D_r}
There exists $\hat s_{\pm}\in S^{0,0}$ such that 
$$\mathcal D^{\pm}u_{\pm}=r^{-2}\Delta_0Op(\hat s_{\pm})u_{\pm}+\nabla\cdot \tilde S_{\pm} u_{\pm} +\hat u_{\pm},$$
where $\Delta_0$ is the Laplace-Beltrami operator on $\mathbb S^{n-1}$, and $\hat u_{\pm}\in L^{2,1/3}(U)$. Here, $\tilde S_{\pm}$ is an operator written as $\tilde S_{\pm}=P_{\perp}\hat\eta(r)Op(\tilde s_{\pm})$, where $\hat \eta$ is a function satisfying \eqref{myeq2.3} with $\eta$ and $\kappa$ replaced by $\hat\eta$ and $2^{-3}\kappa$ respectively, $\tilde s_{\pm}$ is a $n$-dimensional vector whose elements are symbols in $S^{0,-1}$ and $P$ is an orthonormal projection onto the tangent space on the sphere $\mathbb S^{n-1}$, that is, $P_{\perp} A=A-\lvert x\rvert^{-2}x(x\cdot A)$ for any vector $A$.
\end{lem}
\begin{proof}
We shall prove the case of $\mathcal D^+u_+$. The proof for $\mathcal D^-u_-$ is similar.

Let $\eta_1(t)\in C^{\infty}(\mathbb R)$ be a function such that $\eta_1(t)\eta(t)=\eta(t)$ and $\eta_1(t)=0$ for $t<\kappa/2$, where $\kappa$ is as in \eqref{myeq2.3}. Let  $\eta_1$ be the multiplication operator by $\eta_1(\lvert x\rvert)$.

Using the the well-known equality
\begin{equation}\label{myeq5.2.2}
\Delta=\partial_r^2+(n-1)r^{-1}\partial_r+r^{-2}\Delta_0,
\end{equation}
by a straightforward calculation we obtain
\begin{equation}\label{myeq5.0.1}
\mathcal D^-\mathcal D^+=(-\Delta+V-\lambda)+i\partial_r^2K-(n-1)(n-3)/4r^2-r^{-2}\Delta_0+\lvert \nabla Y\rvert^2-\lvert \partial_r Y\rvert^2.
\end{equation}

In the following we denote by $u_j$ functions belonging to $L^{2,1/3}(U)$.

By \eqref{myeq5.0.1} we have
\begin{equation}\label{myeq5.1}
\mathcal D^-\mathcal D^+u_+=r^{-2}\Delta_0\eta_1u_++u_1.
\end{equation}

Let $\hat \chi_{+}$ be supported in $(-2\sqrt{\lambda}/3,\infty)$ and satisfy $\hat \chi_+=1$ on $\mathrm{supp}\, \chi_{+}$. We take $\hat \psi\in C_0^{\infty}(\mathbb R)$ such that $\hat \psi \psi=\psi$.

Let $\eta_{j+1}(t)\in C^{\infty}(\mathbb R),\ j\in \mathbb N$ be  functions such that $\eta_{j+1}(t)\eta_j(t)=\eta_j(t)$ and $\eta_{j+1}(t)=0$ for $t<2^{-j-1}\kappa$. Let $\eta_{j+1}$ be the multiplication operator by $\eta_{j+1}(\lvert x\rvert)$.

Then choosing $\kappa$ large enough, on $\mathrm{supp}\, (\hat\chi_+(\hat x\cdot\xi)\hat \psi(\lvert \xi\rvert^2)\eta_2(\lvert x\rvert))$ the principal symbol of $\eta_3D^-$ is elliptic. Thus, we can construct the parametrix of $\eta_3\mathcal D^-$ there, that is, there exists a symbol $s\in S^{0,0}$ such that $s\#d^-=1+w$ on $\mathrm{supp}\, (\hat\chi_+(\hat x\cdot\xi)\hat \psi(\lvert \xi\rvert^2)\eta_2(\lvert x\rvert))$ where $d^-$ is the symbol of $\eta_3\mathcal D^-$ and $w\in S^{0,-\infty}$ (for the principal symbol, ellipticity and the construction of the parametrix see e.g. \cite{Ho}). We set $\hat s:=s\hat\chi_+(\hat x\cdot\xi)\hat \psi(\xi)\eta_2(\lvert x\rvert)$.

Then applying $Op(\hat s)$ to the both sides of \eqref{myeq5.1} we obtain
\begin{equation}\label{myeq5.2}
\mathcal D^+u_+=Op(\hat s)(r^{-2}\Delta_0\eta_1 u_+)+u_2,
\end{equation}
where $u_2\in L^{2,l}(U)\ l>1/6$.

Set $\nabla_{\perp}:=\nabla-\lvert x\rvert^{-2}x(x\cdot \nabla)$. Then we have
\begin{equation}\label{myeq5.2.1}
\Delta=\partial_r^2+(n-1)r^{-1}\partial_r+\nabla_{\perp}^*\nabla_{\perp}.
\end{equation}

Comparing \eqref{myeq5.2.1} and \eqref{myeq5.2.2}
we can see that $\nabla_{\perp}^*\nabla_{\perp}=r^{-2}\Delta_0$. Thus, a calculation of pseudodifferential operators yields
\begin{equation}\label{myeq5.3}
Op^w(\hat s)(r^{-2}\Delta_0\eta_1 u_+)=r^{-2}\Delta_0Op^w(\hat s)\eta_1u_++\nabla\cdot \tilde S u_+ +u_3,
\end{equation}
where $\tilde S$ is an operator as in the Lemma \ref{D_r}.

By \eqref{myeq5.2} and \eqref{myeq5.3} we obtain the result.
\end{proof}

\begin{lem}\label{sconv1}
We have
$$\int_{\lvert x\rvert=r}(\mathcal D^+u_+)\bar vdS\to 0$$
as $r\to\infty$ where $\int_{\lvert x\rvert=r}fdS$ means the integration of $f$ on $\{x: \lvert x\rvert=r\}$.
\end{lem} 
\begin{proof}
In the following we denote by $F_j(r)$ functions such that $\int_1^{\infty}\lvert F_j(r)\rvert dr<\infty$. We can easily see that the following holds (for similar calculations see e.g. \cite{Is}).
\begin{equation}\label{myeq5.4}
\begin{split}
-i\partial_r\left(\int_{\lvert x\rvert=r}(\mathcal D^+u_+)\bar vdS\right)&=\int_{\lvert x\rvert=r}(\mathcal D^+\mathcal D^+u_+)\bar vdS+F_1(r)\\
&=\int_{\lvert x\rvert=r}(\mathcal D^-\mathcal D^+u_+)\bar vdS-2\sqrt{\lambda}\int_{\lvert x\rvert=r}(\mathcal D^+u_+)\bar vdS\\
&\quad+2\int_{\lvert x\rvert=r}(\partial_rY)(\mathcal D^+u_+)\bar vdS+F_1(r).
\end{split}
\end{equation}

By \eqref{myeq5.0.1} and $\lvert \nabla Y(x)\rvert^2=\mathcal O(\lvert x\rvert^{-2\mu})$, $\lvert \partial_r Y(x)\rvert^2=\mathcal O(\lvert x\rvert^{-2\mu})$ we have
$$\int_{\lvert x\rvert=r}(\mathcal D_r^-\mathcal D_r^+\eta_1u_+)\bar vdS=F_2(r).$$

By Lemma \ref{D_r} we can see that the third term on the right-hand side of \eqref{myeq5.4} is integrable with respect to $r$. As a result, setting
$$\phi(r):=\int_{\lvert x\rvert=r}(\mathcal D^+u_+)\bar vdS,$$
we have $\partial_r\phi(r)=-2i\sqrt{\lambda}\phi(r)+F_3(r)$.

Thus, setting $\phi_1(r)=e^{2i\sqrt{\lambda}}\phi(r)$ there exists a limit $\lim_{r\to\infty}\phi_1(r)$. However, since for some $0<\epsilon<1$ we have $\int_1^{\infty}r^{-\epsilon}\lvert\phi(r)\rvert dr<\infty$, we obtain $\lim_{r\to\infty}\phi_1(r)=0$, and therefore, $\lim_{r\to\infty}\phi(r)=0$
\end{proof}

\begin{lem}\label{asymptotic1}
The limit $\lim\limits_{r\to\infty}\int_{\lvert x\rvert=r} u_+\bar vdS$ exists and the limit is equal to
$$2^{-1}i\lambda^{-1/2}(((\tilde H-\lambda)u,v)-(u,(\tilde H-\lambda)v)),$$
where $(u,v):=\int_{\mathbb R^n}u(x)\bar v(x)dx$.
\end{lem}
\begin{proof}
We have by Green's formula for $\{x: \lvert x\rvert\leq r\}$
\begin{equation}\label{myeq5.4.1}
\begin{split}
\int_{\lvert x\rvert<r}\{(\Delta u_+)\bar v-u_+(\Delta\bar v)\}dx=&\int_{\lvert x\rvert =r}\{(i\mathcal D^+u_+)\bar v-u_+(\overline{i\mathcal D^+v})\}dS\\
&+2i\int_{\lvert x\rvert=r}(\partial_rK)u_+h\bar vdS.
\end{split}
\end{equation}
By Lemma \ref{sconv1} and that $\mathcal D_r^+v=0$ for $\lvert x\rvert$ large enough, the first term on the right-hand side converges to $0$ as $r\to \infty$.

As for the left-hand side we have
\begin{equation}\label{myeq5.5}
\int_{\lvert x\rvert<r}\{(\Delta u_+)\bar v-u_+(\Delta\bar v)\}dx=\int_{\lvert x\rvert<r}\{((\Delta-V+\lambda) u_+)\bar v-u_+((\Delta-V+\lambda)\bar v)\}dx.
\end{equation}
As in the proof of Lemma \ref{sconv1}  by the  the form of $t_+$ we have
\begin{equation}\label{myeq5.6}
(-\Delta+V-\lambda)T_+=\nabla\cdot\tilde S_1+S_0+T_+(-\Delta+V-\lambda),
\end{equation}
where $\tilde S_1$ has the same property as $\tilde S$ in Lemma \ref{D_r} and $S_0:=Op(s_0),\ s_0\in S^{0,-3/2}$. Since we also have $\tilde v:=(-\Delta+V-\lambda)v\in L^{2,l}(\mathbb R^n)$ for some $l>1/2$, the right-hand side of \eqref{myeq5.5} converges to 
\begin{equation}\label{myeq5.7}
\begin{split}
&-(\nabla\cdot\tilde S_1u,v)-(S_0u+T_+\tilde u,v)+(u_+,\tilde v)\\
&\quad=(u,\tilde S_1^*\cdot\nabla_{\perp}v)-(u,S_0v)-(\tilde u,T_+v)+(u,T_+\tilde v),
\end{split}
\end{equation}
where $\tilde u=(-\Delta+V-\lambda)u$ and $(v_1,v_2)$ is the inner product of $v_1$ and $v_2$. Therefore, the limit $\lim\limits_{r\to \infty} 2i\int_{\lvert x\rvert =r}\partial_rKu_+\bar vdS$ exists. 

Moreover, taking the adjoint of \eqref{myeq5.6} we have
$$(-\Delta+V-\lambda)T_+=\tilde S_1^*\cdot\nabla_{\perp}-S_0+T_+(-\Delta+V-\lambda).$$
Thus, the right-hand side of \eqref{myeq5.7} is written as
\begin{equation}\label{myeq5.7.1}
(u,(-\Delta+V-\lambda)T_+v)-(\tilde u,T_+v).
\end{equation}

We can remove $T_+$ in \eqref{myeq5.7.1}, because $(1-T_+)v\in \mathcal S(\mathbb R^n)$ by Lemma \ref{micro2}, and therefore,
$$(u,(-\Delta+V-\lambda)(1-T_+)v)=((-\Delta+V-\lambda)u,(1-T_+)v).$$

Differentiating $\int_{\lvert x\rvert=r}(\partial_rY)u_+\bar vdS$ with respect to $r$ and using Lemma \ref{D_r} we can see that $\lim\limits_{r\to \infty} 2i\int_{\lvert x\rvert =r}(\partial_rY)u_+\bar vdS$ exists. However, since for some $0<\epsilon<1$ we have $\int_{\lvert x\rvert =r}r^{-\epsilon}\lvert \partial_rYu_+\bar v\rvert dSdr<\infty$, the limit is $0$. Thus, taking the limit in \eqref{myeq5.4.1} we obtain
$$\lim\limits_{r\to \infty} 2i\sqrt{\lambda}\int_{\lvert x\rvert =r}u_+\bar vdS=(u,\tilde v)-(\tilde u,v).$$
\end{proof}

\begin{lem}\label{asymptotic2}
$$\lim\limits_{\rho\to \infty}\rho^{-1}\int_{\lvert x\rvert<\rho} u_-\bar vdx=0.$$
\end{lem}
\begin{proof}
We have by a straightforward calculation
\begin{equation}\label{myeq5.8}
\begin{split}
&-i\partial_r\left(\int_{\lvert x\rvert=r}u_-\bar vdS\right)\\
&\quad=\int_{\lvert x\rvert=r}(\mathcal D^-\eta_1u_-)\bar vdS-2\sqrt{\lambda}\int_{\lvert x\rvert =r}u_-\bar vdS+\int_{\lvert x\rvert=r}(\partial_rY)u_-\bar vdS+F(r),
\end{split}
\end{equation}
where $\int_1^{\infty}\lvert F(r)\rvert dr<\infty$.

Setting $\phi(r):=e^{-2i\sqrt{\lambda}r}\int_{\lvert x\rvert=r}u_-\bar vdS$, by \eqref{myeq5.8} and Lemma \ref{D_r} we can see that
$$\int_1^{\infty}r^{-\epsilon}\lvert\partial_r\phi(r)\rvert dr<\infty,$$
for some $0<\epsilon<1$. Thus there exists a constant $C>0$ such that
$$\lvert \phi(r)\rvert<r^{\epsilon}(\int^r_1r_1^{-\epsilon}\lvert \partial_{r_1}\phi(r_1)\rvert dr_1+\lvert \phi(1)\rvert)<Cr^\epsilon.$$
Integrating both sides of  \eqref{myeq5.8} we can see that there exists a constant $\tilde C$ such that
$$\left\lvert\int_{\lvert x\rvert<\rho}u_-\bar vdx\right\rvert\leq \tilde C(\lvert\phi(\rho)\rvert+\rho^{\epsilon})\leq \tilde C(C+1)\rho^{\epsilon},$$
from which the lemma follows.
\end{proof}

\begin{proof}[Proof of Proposition \ref{asymptotic}]
Proposition \ref{asymptotic} follows  immediately from Lemmas \ref{asymptotic1} and \ref{asymptotic2}.
\end{proof}

\section{Uniqueness theorem for nonhomogeneous equations and the outgoing (incoming) property}\label{fifthsec.1}
In the following we use the notations in section \ref{firstsec}, \ref{first.2sec} and \ref{secondsec}.
In this section we introduce the Isozaki's uniquness theorem for nonhomoeneous $N$-body Schr\"odinger operators. First, we need the definition of  a class of symbols of pseudodifferential operators. For $k>0$ and $\tau\in \mathbb R$, we introduce the following.

\begin{dfn}[{\cite[Definition 1.1]{Is2}}]
Let $n\in \mathbb N$. $\mathcal R_{\pm}^k(\tau)$ is the set of $C^{\infty}(\mathbb R^n\times \mathbb R^n)$-functions $p(x,\xi)$ such that
$$\lvert \partial^{\gamma}_x\partial_{\xi}^{\gamma_0}p(x,\xi)\rvert\leq C_{\alpha\beta}\langle x\rangle^{-\lvert\gamma\rvert}\langle \xi\rangle^{-k},$$
for $0\leq \lvert\gamma\rvert\leq k$, $0\leq \lvert\gamma_0\rvert\leq k$ and on $\mathrm{supp}\, p(x,\xi)$
$$\inf_{x,\xi}\pm\hat x\cdot\xi>\pm \tau.$$
\end{dfn}

In Isozaki's uniqueness theorem, outgoing and incoming properties are the conditions of the uniqueness. We define the outgoing and incoming properties as follows.

\begin{dfn}
Let $1/2<l\leq 1$ and $n\in \mathbb N$.
\begin{itemize}
\item[(1)] A function $u\in L^{2,-l}(\mathbb R^n)$ is outgoing (resp., incoming), if there exist $k_0>0$, $0\leq l_0<1/2$ and $\epsilon>0$ such that $Op(p)u\in L^{2,-l_0}(\mathbb R^n)$ for any $p\in \mathcal R_-^{k_0}(\epsilon)$ (resp., $p\in \mathcal R_+^{k_0}(-\epsilon)$).
\item[(2)] A function $u\in L^{2,-l}(\mathbb R^n)$ is strictly outgoing (resp., strictly incoming), if there exists $k_0>0$ such that for any $\epsilon>0$ there exists  $0\leq l_0<1/2$ satisfying the following condition:
$Op(p)u\in L^{2,-l_0}(\mathbb R^n)$ for any $p\in \mathcal R_-^{k_0}(1-\epsilon)$ (resp., $p\in \mathcal R_+^{k_0}(1+\epsilon)$).
\end{itemize}
\end{dfn}

\begin{rem}
Since we can assume $X_a=\mathbb R^n$ for some $n\in \mathbb N$, we can define the outgoing and incoming properties for $\tilde v\in L^{2,-l'}(X_a),\ 1/2<l\leq 1$ in the same way as above. When we consider operators for $X_a$, we write as $\mathcal R_{+,a}^{k}(\tau)$.
By Lemma \ref{micro1}, we can see that $\tilde R_a(\lambda+i0)v:=(\tilde H_a-\lambda-i0)^{-1}v$ (resp., $\tilde R_a(\lambda-i0)v:=(\tilde H_a-\lambda+i0)^{-1}v$), $v\in L^{2,l}(X_a),\ l>1/2$ is strictly outgoing (resp., strictly incoming), where $\tilde H_a:=-\Delta_a+\tilde I_a$.
\end{rem}

The outgoing and incoming properties can be written using the Graf's vector field. Let us introduce the following class of functions.

\begin{dfn}[\cite{GIS}]
\begin{itemize}
\item[(1)] Let $\mathcal V$ be the set of $C^{\infty}(X)$-functions $v$ on $X$ such that for any $\alpha\in \mathbb N^{\mathrm{dim}\, X}$ and $k\in \mathbb N$ there exists $C_{\alpha,k}$ satisfying the following inequality:
$$\lvert \partial_x^{\alpha}(x\cdot\nabla)^kv(x)\rvert\leq C_{\alpha,k}.$$
\item[(2)] Let $\mathcal V_+^1$ be the set of  positive $C^{\infty}(X)$-functions $r$ on $X$ such that
$$r(x)^2-\lvert x\rvert^2\in \mathcal V.$$
\end{itemize}
\end{dfn}

We need the following differential operator.

\begin{lem}[{\cite[Lemma 2.1]{GIS}}]
Let $\lambda\in \sigma_{ess}(H)\setminus(\sigma_{pp}(H)\cap\mathcal T(H))$ and $\epsilon>0$ be given. Then there exist an open neighborhood $N_{\lambda}$ of $\lambda$ and $r\in \mathcal V_+^1$ such that with $A$ given as the self-adjint operator on $\mathcal  H:=L^2(X)$ by
$$A:=(\omega\cdot D+D\cdot \omega)/2,\ \omega=r\nabla r,\ D:=-i\nabla,$$
\begin{itemize}
\item[(1)] $i[H,A]$ defined as a form on $\mathcal D(H)\cap \mathcal D(A)$ extends to a symmetric operator on $\mathcal D(H)$, and in fact
$$i[H,A]=\sum_{a\in \mathcal A}v_{a}\nabla^a,\ v_a\in \mathcal V.$$
\item[(2)] $\varphi(H)i[H,A]\varphi(H)\geq 2d(\lambda)(1-\epsilon)\varphi(H)^2$ for all real-valued $\varphi\in C_0^{\infty}(N_{\lambda})$, where $d(\lambda):=\inf\{\lambda-t:t\in\mathcal T(H),t<\lambda\}$.
\end{itemize}
\end{lem}

Let $Y$ be the operator of multiplication by $Y(x):=\langle x\rangle$ on $\mathcal H$.

\begin{dfn}\label{Grafvector}
With $B$ given as the self-adjoint operator on $\mathcal H$
$$B:=r^{-1/2}Ar^{-1/2}=(\nabla \cdot D+D\cdot\nabla r)/2,$$
we let $\mathcal D$ be the domain
$$\mathcal D:=\bigcap\mathcal D(Q),$$
where the intersection is over all polynomials $Q$ in $Y$ and $B$.
\end{dfn}

\begin{dfn}
We define for any $m\in \mathbb R$ the class $\mathcal Op^m(Y)$ of operators $P$ with the properties
\begin{itemize}
\item[1)] $\mathcal D(P)$ and $\mathcal D(P^*)$ contain $\mathcal D$, and $P$ and $P^*$ restricted to $\mathcal D$ map into itself.
\item[2)] For any $n\in\mathbb N,\ \alpha,\beta\in\mathbb R$ such that $\alpha+\beta=n-m$, $Y^{\alpha}\mathrm{ad}_n(P,B)Y^{\beta}$ extends to a bounded operator on $\mathcal H$.
\end{itemize}
\end{dfn}
Here $\mathrm{ad}_0(P,B):=P$, $\mathrm{ad}_n(P,B):=[\mathrm{ad}_{n-1}(P,B),B],\ n\geq 1$.

Let for any $m\in \mathbb R$, $\mathcal F^m$ be the class of $C^{\infty}$-functions on $\mathbb R$ such that
$$\lvert f^{(k)}(t)\rvert\leq C_k(1+\lvert t\rvert)^{m-k},\ \forall k\geq 0.$$ As in \cite[Lemma 2.3]{GIS}, we have the following lemma.

\begin{lem}\label{Op}
If $f\in \mathcal F^m$ for some $m<0$ and $a\in\mathcal A$, we have $f(H^a)\in \mathcal Op^0(Y)$.
\end{lem}

\begin{lem}[\cite{Is3}]\label{AB}
Let $A$ and $B$ be lower-semibounded self-adjoint operators, $\mathcal D(A)=\mathcal D(B)$, and $(A-B)(A+i)^{-1}$ is a bounded operator. Then, we have
$$\lVert F(A<R_1)F(B>R_2)\rVert=\mathcal O(R_2^{-1/2}),$$
as $R_2\to \infty$, where $F(t\gtrless r)$ is a function such that $F(t\gtrless r)=1$ for $t\gtrless r$ and $F(t\gtrless r)=0$ for $t\lessgtr r$ for $r\in \mathbb R$.
\end{lem}

\begin{proof}
Set
$$G(t,R_2):=e^{-tB}F(B>R_2)F(A<R_1)F(B>R_2)e^{-tB}.$$
Then, there exists $C>0$ such that
\begin{equation}\label{myeq fifthsec.1.1}
\begin{split}
-\frac{d}{dt}G(t,R_2)=&2\mathrm{Re}\, e^{-tB}F(B>R_2)(B-A)F(A<R_1)F(B>R_2)e^{-tB}\\
&+2e^{-tB}F(B>R_2)(B-A)AF(A<R_1)F(B>R_2)e^{-tB}\\
&\leq Ce^{-2tR_2}.
\end{split}
\end{equation}
Noticing $G(0,R_2)=F(B>R_2)F(A<R_1)F(B>R_2)$ and integrating \eqref{myeq fifthsec.1.1} with respect to $t$, we obtain the result.
\end{proof}

Set
\begin{align*}
\mathcal F^m_+(a)&:=\{f\in \mathcal F^m : \mathrm{supp}\, f\subset (a,\infty)\},\\
\mathcal F^m_-(a)&:=\{f\in \mathcal F^m : \mathrm{supp}\, f\subset (-\infty,a)\}
\end{align*}
As in the proof of \cite[Theorem 2.12]{GIS}, we have the following lemma.

\begin{lem}\label{vectorpseudo}
Let $\tau\in \mathbb R$. Then for any $F_{\pm}\in \mathcal F^0_{\pm}(\tau),\ Op(p_{\pm}),\ p_{\pm}\in\mathcal R_{\pm}(\tau)$ and $s>0$ we have
$$X^sP_{\mp}F_{\pm}(B)X^s\in\mathcal L(\mathcal H).$$
\end{lem}

The outgoing and incoming properties can be stated using $B$ in Definition \ref{Grafvector}.

\begin{lem}
Let $0\leq l_0<1/2<l\leq 1$, $n\in\mathbb N$ and $u\in L^{2,-l}(\mathbb R^n)$. Assume that there exists $q\in C_0^{\infty}(\mathbb R)$ such that $(1-q(-\Delta))u\in L^{2,-l_0}(\mathbb R^n)$. Then $u$ is outgoing (resp., incoming) if and only if there exists $\epsilon>0$ and $0\leq l_1<1/2$ such that $F_{-}(B)u\in L^{2,-l_1}(\mathbb R^n)$ (resp., $F_{+}(B)u\in L^{2,-l_1}(\mathbb R^n)$) for any $F_-\in \mathcal F^0_-(\epsilon)$ (resp., $F_+\in \mathcal F^0_+(-\epsilon)$).
\end{lem}

\begin{proof}
We prove that there exists $\epsilon'>0$ such that $Op(p)u\in L^{2,-l_0}(\mathbb R^n)$ for any $p\in \mathcal R^{k_0}_-(\epsilon')$, assuming that $F_{-}(B)u\in L^{2,-l_1}(\mathbb R^n)$ for any $F_-\in \mathcal F^0_-(\epsilon)$ and $k_0>0$. The proofs for the incoming case and the converse statements are similar.

Let $\epsilon'>0$ satisfy $\epsilon'<\epsilon$ and let $\sigma>0$ be a number such that $\epsilon'+3\sigma<\epsilon$. Let $f_+\in\mathcal F_+^0(\epsilon'+\sigma)$ and $f_-\in\mathcal F_-^0(\epsilon'+2\sigma)$ satisfy $f_+(t)+f_-(t)=1$. Then we have
$$Op(p)u=Op(p)f_+(B)u+Op(p)f_-(B)u.$$

By Lemma \ref{vectorpseudo} we have $Op(p)f_+(B)u\in L^{2,-l_1}(\mathbb R^n)$. By the assumption we can also see that $Op(p)f_-(B)u\in L^{2,-l_1}(\mathbb R^n)$ which completes the proof.
\end{proof}

The following lemma is the Isozaki's uniqueness theorem.

\begin{lem}[{\cite[Theorem 1.3]{Is2}}]\label{Iunique}
Let $1/2<l\leq1$ and $\lambda\in \sigma_{ess}(H)\setminus(\sigma_p(H)\cup\mathcal T(H))$. Suppose that $u\in L^{2,-l}(X)$ satisfies $(H-\lambda)u=0$ and $u$ is outgoing or incoming. Then $u=0$.
\end{lem}

The following Lemma is useful to confirm the outgoing and incoming properties.

\begin{lem}\label{outdirect}
Let $0\leq l_0<1/2<l\leq 1$ and $u_{\alpha}(x^a)$ be an eigenfunction of $H^a$ corresponding to a non-threshold channel $\alpha$. If $v\in L^{2,-l}(X_a)$ is outgoing (resp., incoming) and there exists $q\in C_0^{\infty}(\mathbb R)$ such that $(1-q(-\Delta_a))v\in L^{2,-l_0}(X_a)$. Then $(J_{\alpha}v)(x)$ is outgoing.
\end{lem}

\begin{proof}
We prove only the outgoing case. The incoming case is proved in the same way. 

Let $k_0>0$, $\epsilon>0$, $0\leq l_1<1/2$ and $Op(p)v\in L^{2,-l_1}(X_a)$ for any $p\in \mathcal R_{+,a}^{k_0}(\epsilon)$. Let $\epsilon'>0$ a number such that $\epsilon'<\epsilon$, and $\sigma>0$ a number such that $\epsilon'+4\sigma<\epsilon$. We shall prove that $F_-(B)(J_{\alpha}v)\in L^{2,-l_2}(X)$ for any $F_-\in \mathcal F_-(\epsilon')$, where $l_2:=\max\{l_0,l_1\}$. Let $\varphi\in C_0^{\infty}(\mathbb R)$ be  a function such that $\varphi(t)=1$ near $E_{\alpha}$, where $E_{\alpha}$ is the eigenvalue corresponding to the channel $\alpha$. Then we have $\varphi(H^a)u_{\alpha}=u_{\alpha}$. Let $\psi\in C_0^{\infty}(\mathbb R)$ be a function such that $\psi\varphi=\varphi$. Let $f\in C^{\infty}(\mathbb R)$ satisfy $f(t)=1$ for $t>2$ and $f(t)=0$ for $t<1$.

By Lemma \ref{AB}, for $C>0$ large enough $K_{C}:=f(-\Delta^a/C)\psi(H^a)$ on $L^2(X^a)$ satisfies $\lVert K_C\rVert<1/2$. Setting $f_0:=1-f$ we have $\varphi(H^a)=(K_C+f_0(-\Delta^a/C))\varphi(H^a)$, and therefore, $\varphi(H^a)=(1-K_C)^{-1}f_0(-\Delta^a/C)\varphi(H^a)$. Thus we obtain
$$u_{\alpha}=(1-K_C)^{-1}f_0(-\Delta^a/C)u_{\alpha}.$$

We denote by $w_j$ functions such that $w_j\in L^{2,-l_2}(X)$.

Since by Lemma \ref{Op} we have $K_C\in \mathcal Op^0(Y)$, we have $(1-K_C)^{-1}\in \mathcal Op^0(Y)$. Therefore, there exist $F_{j,-}'\in \mathcal F_-(\epsilon'),\ j=1,2$ such that
\begin{equation}\label{myeq fifthsec.1.2}
F_-(B)(1-K_R)^{-1}f_0(-\Delta^a/R)(J_{\alpha}v)=\sum_{j=1}^2K'_jF_{j,-}'(B)f_0(-\Delta^a/R)(J_{\alpha}v)+w_1.
\end{equation}

Let us prove that the first term in the right-hand side of \eqref{myeq fifthsec.1.2} belongs to $L^{2,-l_2}(X)$. By the assumption we can see that
$$u_{\alpha}v=q(-\Delta_a)(J_{\alpha}v)+w_2,$$
where $q\in C_0^{\infty}(\mathbb R)$ is as in the assumption.

We can easily see that there exist $p_+\in \mathcal R^{k_0}_+(\epsilon'+\sigma)$ and $p_-\in \mathcal R^{k_0}_-(\epsilon'+2\sigma)$ such that the following holds: there exits $C'>0$ such that $\lvert \xi\rvert<C'$ on $\mathrm{supp}\, p_{\pm}$, and
$$f_0(-\Delta^a/C)q(-\Delta_a)(J_{\alpha}v)=Op(p_+)(J_{\alpha}v)+Op(p_-)(J_{\alpha}v)+w_3.$$

By Lemma \ref{vectorpseudo} we have $F_{j,-}'(B)Op(p_+)(J_{\alpha}v)\in  L^{2,-l_2}(X)$. Thus, we only need to prove $Op(p_-)(J_{\alpha}v)\in  L^{2,-l_2}(X)$, in order to prove that the right-hand side of \eqref{myeq fifthsec.1.2} belongs to $L^{2,-l_2}(X)$. 

Let $\chi\in C^{\infty}(X)$ be a homogeneous function of degree $0$ for $\lvert x\rvert >1$ such that the following holds: there exists $\delta>0$ such that $\chi(x)=1$ for $x\in \{x: \lvert x^a\rvert>2\delta\lvert x\rvert,\ \lvert x\rvert>1\}$ and $\chi(x)=0$ for $x\in \{x: \lvert x^a\rvert<\delta\lvert x\rvert, \ \lvert x\rvert>1\}$. Then it is easy to see that
$$Op(p_-)(J_{\alpha}v)=Op(p_-(x,\xi)\chi(x))(J_{\alpha}v)+w_4.$$

Choosing $\delta$ sufficiently small, we can assume $(x^a\cdot\xi^a)/\lvert x\rvert<\sigma$ on
$$\mathrm{supp}\,  (p(x,\xi)\chi(x)).$$
Thus, there exists $\tilde p_-(x_a,\xi_a)\in \mathcal R_-^{k_0}(\epsilon'+3\sigma)$ such that
$$Op(p_-(x,\xi)\chi(x))(J_{\alpha}v)=Op(p_-(x,\xi)\chi(x))Op(\tilde p_-)(J_{\alpha}v)+w_5.$$

By the assumption that $Op(p)v\in L^{2,-l_2}(X_a)$ for any $p\in \mathcal R_{+,a}^{k_0}(\epsilon)$, we have $Op(\tilde p_-)v\in L^{2,-l_2}(X_a)$, and therefore, $Op(\tilde p_-)(J_{\alpha}v)\in L^{2,-l_2}(X)$. Thus, we obtain $Op(p_-)(J_{\alpha}v)\in  L^{2,-l_2}(X)$ which completes the proof.
\end{proof}

\section{Scattering matrices and generalized Fourier transforms}\label{fifth.1sec}
In the following we use the notations in section \ref{firstsec}, \ref{first.2sec} and \ref{secondsec}.
In this section we define the scattering matrices and the generalized Fourier transforms.

\begin{lem}\label{B^*}
Let $\alpha$ be a non-threshold channel and $u\in \mathcal B^*(X)$. Then, there exists a constant $C$ such that
$$\lVert \pi_{\alpha}u\rVert_{\mathcal B^*(X_a)}\leq C\lVert u\rVert_{\mathcal B^*(X)}.$$

\end{lem}
\begin{proof}
There exist $C_1,C_2,C_3>0$ such that
\begin{align*}
\lVert \pi_{\alpha}u\rVert_{\mathcal B^*(X_a)}^2&\leq C_1 \sup_{\rho>1}\rho^{-1}\int_{\lvert x_a\rvert<\rho}\int \langle x^a\rangle^{-2}\lvert u(x)\rvert^2dx^adx_a\\
&\leq C_2 \sup_{\rho>1}\rho^{-1}\sup_{j\geq 0} \rho_j^{-1}\int_{\lvert x_a\rvert<\rho}\int_{x^a\in\Omega_j} \lvert u(x)\rvert^2dx^adx_a\\
&\leq C_2\sup_{\rho>1,j\geq 0}\rho^{-1}\rho_j^{-1}\int_{\substack{\lvert x_a\rvert<\rho\\
x^a\in \Omega_j}}\lvert u(x)\rvert^2dx\\
&\leq C_2\sup_{\rho>1,j\geq 0}(\rho+\rho_j)^{-1}\int_{\lvert x\rvert <2(\rho+\rho_j)}\lvert u(x)\rvert^2dx\leq C_3\lVert u\rVert_{\mathcal B^*(X)}^2.
\end{align*}
\end{proof}

Let $\lambda>0$, $u\in \mathcal B^*(X)$, $h\in C_0^{\infty}(C_a')$ and $\alpha$ be a non-threshold channel. Then, by Lemma \ref{B^*} there exist $C,C'>0$ such that
\begin{equation}\label{myeq5.8.0.1}
\begin{split}
&\left \lvert\lim_{\rho\to\infty}\rho^{-1}\int_{\lvert x_a\rvert<\rho}e^{\mp iK_a(x_a,\lambda_{\alpha})}r_a^{-(n_a-1)/2}h(\hat x_a)(\pi_{\alpha}u)(x_a)dx_a\right\rvert\\
&\qquad\leq C\lVert \pi_{\alpha}u\rVert_{\mathcal B^*(X_a)}\lVert h\rVert_{L^2(C_a')}\\
&\qquad\leq C'\lVert u\rVert_{\mathcal B^*(X)}\lVert h\rVert_{L^2(C_a')}.
\end{split}
\end{equation}
By Proposition \ref{asymptotic} and \eqref{myeq5.8.0.1} we can define a distribution $Q^{\pm}_{\alpha}(u)$ on $C_a'$ by
\begin{align*}
(&Q^{\pm}_{\alpha}(u))(h)\\
&=\lim_{\rho\to\infty}\rho^{-1}\int_{\lvert x_a\rvert<\rho}e^{\mp iK_a(x_a,\lambda_{\alpha})}r_a^{-(n_a-1)/2}h(\hat x_a) (\pi_{\alpha}u)(x_a)dx_a,
\end{align*}
for any $h\in C_0^{\infty}(C_a')$ and by Riesz theorem we can see that $Q_{\alpha}^{\pm}(u)\in L^2(C_a')$. Since the Lebesgue measure of $C_{a,\mathrm{sing}}$ is $0$, we can extend $Q_{\alpha}^{\pm}(u)$ to $C_a$ so that $Q_{\alpha}^{\pm}(u)\in L^2(C_a)$.

Set
$$(\hat v_{\alpha}^{\pm}(\lambda))(x_a)=\hat v_{\alpha,\pm}(\lambda,x_a):=\eta(r_a)\bar h(\hat x_a)r_a^{-(n_a-1)/2}e^{\pm iK_a(x_a,\lambda_{\alpha})},$$
with $\eta$ as in \eqref{myeq2.3} for some $\kappa>0$.
Applying Proposition \ref{asymptotic} we have
\begin{equation}\label{myeq5.8.1}
\begin{split}
\lim_{\rho\to\infty}\rho^{-1}&\int_{\lvert x_a\rvert<\rho}e^{\mp iK_a(x_a,\lambda_{\alpha})}r_a^{-(n_a-1)/2}h(\hat x_a)(\pi_{\alpha}u)(x_a)dx_a\\
&=2^{-1}i\lambda_{\alpha}^{-1/2}\{((\tilde H_a-\lambda_{\alpha})(\pi_{\alpha}u),v_{\alpha,\pm})_a\\
&\quad-(\pi_{\alpha}u,(\tilde H_a-\lambda_{\alpha})v_{\alpha,\pm})_a\},
\end{split}
\end{equation}
where $\tilde H_a:=-\Delta_a+\tilde I_a$ and $(v_1,v_2)_a:=\int_{X_a}v_1(x_a)\bar v_2(x_a)dx_a$.

Let $u\in \mathcal B^*(X)$ be a generalized eigenfunction of $H$ with an eigenvalue $\lambda$. Then by \eqref{myeq5.8.1} we have
\begin{equation}\label{myeq5.9}
\begin{split}
\lim_{\rho\to\infty}\rho^{-1}&\int_{\lvert x_a\rvert<\rho}\int_{X^a}e^{\mp iK_a(x_a,\lambda_{\alpha})}r_a^{-(n_a-1)/2}h(\hat x_a)(\pi_{\alpha}u)(x_a)dx_a\\
&=2^{-1}i\lambda_{\alpha}^{-1/2}\{((\tilde H_a-\lambda_{\alpha})(\pi_{\alpha}u),v_{\alpha,\pm})_a-(\pi_{\alpha}u,(\tilde H_a-\lambda_{\alpha})v_{\alpha,\pm})_a\}\\
&=2^{-1}i\lambda_{\alpha}^{-1/2}\{(\pi_{\alpha}((\tilde I_a-I_a)u),v_{\alpha,\pm})_a-(\pi_{\alpha}u,(\tilde H_a-\lambda_{\alpha})v_{\alpha,\pm})_a\}\\
&=-2^{-1}i\lambda_{\alpha}^{-1/2}(u,(H-\lambda)J_{\alpha}v_{\alpha,\pm}),
\end{split}
\end{equation}
where $(v_1,v_2):=\int_{X}v_1(x)\bar v_2(x)dx$, and we used that $(H-\lambda)u=0$ in the second equality.

Let $\alpha$ and $\beta$ be non-threshold channels. Now we define the scattering matrix as the map
$$\Sigma_{\beta\alpha}(\lambda):C_0^{\infty}(C_a')\to L^2(C_b),$$
given by
$$\Sigma_{\beta\alpha}(\lambda)g:=Q^+_{\beta}(P_{\alpha,+}(\lambda)g).$$

Next we shall define the generalized Fourier transforms.  Let $f\in L^{2,l}(X)$ for some $l>1/2$. In the similar way as \eqref{myeq5.9} we obtain
\begin{equation}\label{myeq5.10}
\begin{split}
\lim_{\rho\to\infty}\rho^{-1}&\int_{\lvert x_a\rvert<\rho}e^{\mp iK_a(x_a,\lambda_{\alpha})}r_a^{-(n_a-1)/2}h(\hat x_a)(\pi_{\alpha}R(\lambda\pm i0)f)(x_a)dx_a\\
&=2^{-1}i\lambda_{\alpha}^{-1/2}\{(f,J_{\alpha}v_{\alpha,\pm})-(R(\lambda\pm i0)f,(H-\lambda)J_{\alpha}v_{\alpha,\pm})\},
\end{split}
\end{equation}

For a non-threshold channel $\alpha$ we define the generalized Fourier transform as the map
$$\mathcal G_{\alpha}^{\pm}(\lambda): L^{2,l}(X)\to L^2(C_a),$$
given by
$$\mathcal G_{\alpha}^{+}(\lambda)f=D_{\alpha}^+(\lambda)Q_{\alpha}^{\pm}(R(\lambda\pm i0)f),$$
and 
$$\mathcal G_{\alpha}^{+}(\lambda)f=D_{\alpha}^-(\lambda)\mathcal RQ_{\alpha}^{\pm}(R(\lambda\pm i0)f),$$
where $(\mathcal Rg)(\hat x_a):=g(-\hat x_a)$ and
$$D_{\alpha}^{\pm}(\lambda):=\pi^{-1/2}\lambda_{\alpha}^{1/4}e^{\pm i\pi(n_a-3)/4}.$$

\section{Asymptotic behaviors of functions in the range of the resolvent and Poisson operators}\label{fifth.2sec}
In the following we use the notations in section \ref{firstsec}, \ref{first.2sec}, \ref{secondsec}, \ref{fourthsec}, \ref{fifthsec.1} and \ref{fifth.1sec}.
In this section we study the $Q_{\beta}^{\pm}(u)$, where $u=P_{\alpha,\pm}g$ for $g\in C_0^{\infty}(C_a')$ or $u=R(\lambda+i0)f$ for $f\in L^{2,l}(X)$, $l>1/2$.

Set for $\epsilon>0$, $Y_{a,\epsilon}:=Y_a\cap\{x \in X: \lvert x^a\rvert<\epsilon\lvert x\rvert\}$, where $Y_a$ is defined as in \eqref {myeqfirst.2.0.1.0}. Let $J_{a,\epsilon}\in C^{\infty}(X)$ be homogeneous of degree zero outside $\{x\in X: \lvert x\rvert=1\}$ and $\mathrm{supp}\, J_{a,\epsilon}\in Y_{a,\epsilon}$. Set also $d(\lambda):=\inf\{\lambda-t: t\in\mathcal T(H),\ t<\lambda\}$.
We need the following lemma.
\begin{lem}[{\cite[Theorem 3.5]{GIS}}]\label{conicmicro}
Let $\lambda\in (\inf\sigma(H^a),\infty)\setminus\mathcal T(H)$, where $\sigma(A)$ is the spectra of $A$. Then, for any $\tau<\sqrt{d(\lambda)}$ (resp., $\tau>-\sqrt{d(\lambda)}$) there exist $\epsilon>0$, $C>0$ and a neighborhood $N_{\lambda}$ of $\lambda$ such that the following holds: for any $m\in \mathbb N$, $s'>s>m-1/2$, $p_-\in \mathcal R_{-,a}(\tau)$ (resp., $p_+\in \mathcal R_{+,a}(\tau)$) and $J_{a,\epsilon}$ we have
$$\lVert Y^{s-m}Op(p_-)J_{a,\epsilon}R(z)^mY^{-s'}\rVert_{\mathcal L(\mathcal H)}\leq C,$$
$$(resp., \lVert Y^{s-m}Op(p_+)J_{a,\epsilon}R(z)^mY^{-s'}\rVert_{\mathcal L(\mathcal H)}\leq C,)$$
uniformly in $\mathrm{Re}\, z\in N_{\lambda}$ and $\mathrm{Im}\, z>0$ (resp., $\mathrm{Im}\, z<0$).
\end{lem}

\begin{pro}\label{resolventasympt}
Let $f\in L^{2,l}(X)$, $l>1/2$ and $\lambda\in\sigma_{ess}(H)\setminus(\sigma_{pp}(H)\cap\mathcal T(H))$. Then, for any non-threshold channel $\alpha$ such that $\lambda>E_{\alpha}$, we have $Q_{\alpha}^-(R(\lambda+i0)f)=0$ and $Q_{\alpha}^+(R(\lambda-i0)f)=0$.
\end{pro}

\begin{proof}
We shall prove only the case of $Q_{\alpha}^-(R(\lambda+i0)f)$. The case of $Q_{\alpha}^+(R(\lambda-i0)f)$ is similar.

Set $u:=R(\lambda+i0)f$. For some $l'$ and any $\epsilon>0$ we have $(\tilde H_a-\lambda_{\alpha})(\pi_{\alpha}u)\in L^{2,l'}(Z_a^{\epsilon})$, where $Z_a^{\epsilon}$ is defined as \eqref{myeqfirst.2.0.1.1}.
Let $\check \chi_{\pm,\lambda_{\alpha}}$ be functions such that 
$\check \chi_{-}(s)=1$ for  $s<\sqrt{d(\lambda)}/3$, $\check \chi_{-}(s)=0$ for $s>2\sqrt{d(\lambda)}/3$, and $\check \chi_{+}+\check \chi_{-}=1$, We define $\check t_{\alpha}^{\pm}$ by $\check t_{\alpha}^{\pm}:=\check \chi_{\pm}(\hat x_a\cdot\xi_a)\eta(r_a)  \psi_{\alpha}(\lvert \xi_a\rvert^2)$, where $\psi_{\alpha}=1$ near $\lambda_{\alpha}$ and $\eta$ is as in \eqref{myeq2.3} for some $\kappa>0$.

Define $u_{\pm}:=\check T_{\alpha}^{\pm}(\pi_{\alpha}u)$. Then, as in the proof of Proposition \ref{asymptotic} we obtain
\begin{align*}
\lim_{\rho\to\infty}\rho^{-1}\int_{\lvert x_a\rvert<\rho}&e^{ iK_a(x_a,\lambda_{\alpha})}r_a^{-(n_a-1)/2}h(\hat x_a)\\
&\cdot(u(x_a)-u_+(x_a)-u_-(x_a))dx_a=0,\ h\in C_0^{\infty}(C_a')
\end{align*}
and
$$\lim_{\rho\to\infty}\rho^{-1}\int_{\lvert x_a\rvert<\rho}e^{ iK_a(x_a,\lambda_{\alpha})}r_a^{-(n_a-1)/2}h(\hat x_a)u_+(x_a)dx_a=0,\ h\in C_0^{\infty}(C_a').$$

As for $u_-$, for any $\epsilon_1>0$ let $\chi_1,\chi_2\in C^{\infty}(X)$ be functions such that $\chi_1(x)=1$ on $\{x\in X: \lvert x\rvert>1,\ \lvert x^a\rvert<\epsilon_1\lvert x_a\rvert\}$, $\chi_1(x)=0$ on $\{x\in X: \lvert x\rvert>1,\ \lvert x^a\rvert>2\epsilon_1\lvert x_a\rvert\}$, and $\chi_1+\chi_2=1$. We also denote by $\chi_j$ the operator of multiplication  by $\chi_j(x)$.

Then, we have
$$u_-=\pi_{\alpha}(\chi_1\check T_{\alpha}^-u)+\pi_{\alpha}(\chi_2\check T_{\alpha}^-u).$$

By the exponential decay of $u_{\alpha}$ we can see that $\pi_{\alpha}(\chi_2\check T_{\alpha}^-)u\in L^{2,\tilde l}(X_a)$ for any $\tilde l>0$.

As for $\pi_{\alpha}(\chi_1\check T_{\alpha}^-u)$ by Lemma \ref{conicmicro} for sufficiently small $\epsilon_1$ we have
$$\pi_{\alpha}(\chi_1\check T_{\alpha}^-u)\in L^{2,-l}(X_a),$$
for some $0\leq l<1/2$. Therefore, we have $Q_{\alpha}^-(u)=0$
\end{proof}

We can see that for $g\in C_0^{\infty}(C_a')$ the incoming wave of $P_{\alpha,+}(\lambda)g$ consists only of the wave from the channel $\alpha$.
\begin{pro}
Let $\alpha$ and $\beta$ be non-threshold channels and $\alpha\neq\beta$. Then, $Q_{\alpha}^-(P_{\alpha,+}(\lambda)g)=g$ and $Q_{\beta}^-(P_{\alpha,+}(\lambda)g)=0$.
\end{pro}
\begin{proof}
By Proposition \ref{resolventasympt}, the part  $R(\lambda\pm i0)(H-\lambda)(J_{\alpha}v_{\alpha}^{\mp})$ of $P_{\alpha,\pm}(\lambda)g$ does not contribute to $Q_{\alpha}^{\mp}(P_{\alpha,\pm}(\lambda)g)$.

Since the remaining part is $J_{\alpha}v_{\alpha}^{\mp}$, we can easily see that the proposition holds.
\end{proof}

\section{Equivalence of the scattering matrices}\label{sixthsec}
In the following we use the notations in section \ref{firstsec}, \ref{first.2sec}, \ref{secondsec} and \ref{fifth.1sec}.
In this section we prove Theorem \ref{sme}.

\begin{proof}[Proof of Theorem \ref{sme}]
Let $f_1\in C^{\infty}(X_a)$ and $f_2\in C^{\infty}(X_b)$ satisfy $\hat f_1\in C_0^{\infty}(X_a)$, $\hat f_2\in C_0^{\infty}(X_b)$, $\hat f_1(\lambda,\cdot)\in C_0^{\infty}(C_a')$ for any $\lambda$ such that $\mathrm{supp}\, \hat f_1(\lambda,\cdot)\neq \varnothing$, and $\hat f_2(\lambda,\cdot)\in C_0^{\infty}(C_b')$ for any $\lambda>0$ such that $\mathrm{supp}\, \hat f_2(\lambda,\cdot)\neq \varnothing$. Here $(\hat f_1(\lambda))(\omega)=\hat f_1(\lambda,\omega):=(F_{\alpha}f_1)(\lambda,\omega),\ \omega\in C_a$ and $(\hat f_2(\lambda))(\omega_0)=\hat f_2(\lambda,\omega_0):=(F_{\beta}f_2)(\lambda,\omega_0),\ \omega_0\in C_b$, where $F_{\alpha}$ is defined by \eqref{myeq1.2.5}.

We only need to prove
\begin{equation}\label{myeq6.0.0.1}
\begin{split}
&\int_{\max\{E_{\alpha},E_{\beta}\}}^{\infty}(\hat S_{\beta\alpha}(\lambda)\hat f_1(\lambda),\hat f_2(\lambda))_{\hat K_{b}}d\lambda\\
&\quad=\int_{\max\{E_{\alpha},E_{\beta}\}}^{\infty}e^{i\pi(n_a+n_b-2)/4}\lambda_{\alpha}^{1/4}\lambda_{\beta}^{-1/4}(\Sigma_{\beta\alpha}(\lambda)\mathcal R\hat f_1(\lambda),\hat f_2(\lambda))_{\hat K_{b}}d\lambda,
\end{split}
\end{equation}
where $K_b:=L^2(C_b)$ and $(v_1,v_2)_M$ is the inner product of $v_1$ and $v_2$ in a Hilbert space $M$.

(i) First, we consider the case such that $\alpha=\beta$.

Let $\varphi\in C_0^{\infty}(C_b')$ be a function satisfying $\varphi(\hat{\xi}_b)\hat f_2(\lambda,\hat{\xi}_a)=\hat f_2(\lambda,\hat{\xi}_a)$ for any $\lambda$ such that $\mathrm{supp}\, \hat f_2(\lambda,\cdot)\neq 0$, and $\psi\in C^{\infty}_0(\mathbb R)$ be a function satisfying $\psi(t)=1$ for any $t>0$ such that $\mathrm{supp}\, \hat f_2(\sqrt{t},\cdot)\neq \varnothing$. Let $\eta(t)\in C^{\infty}(\mathbb R)$ be a function as in \eqref{myeq2.3} for some $\kappa>0$, and $\theta_+(t)\in C^{\infty}(\mathbb R)$ be a function such that $\theta_+(t)=0$ for $t<1-2\epsilon$ and $\theta_+(t)=1$ for $t>1-\epsilon$, where $1>\epsilon>0$ is determined later.

Set $D_a:=-i\nabla_a$. Then, we have
\begin{equation}\label{myeq6.0.1}
\begin{split}
W_{\beta}^+f_2&=s-\lim_{t\to+\infty}e^{itH}\varphi(D_b)\psi(\lvert D_b\rvert^2)J_{\beta}e^{-i(S_{\beta,\pm}(D_b,t)+\lambda_{\beta}t)}f_2\\
&=s-\lim_{t\to+\infty}e^{itH}T^+_{b}J_{\beta}e^{-i(S_{\beta,\pm}(D_b,t)+\lambda_{\beta}t)}f_2\\
&=s-\lim_{t\to+\infty}e^{itH}T^+_{b}J_{\beta}e^{-ith_{\beta}}\tilde W_{\beta}^+f_2\\
&=T_{\beta}^+J_{\beta}\tilde W_{\beta}^+f_2+i\int_0^{\infty}e^{isH}G_{\beta}^+e^{-is h_{\beta}}\tilde W_{\beta}^+f_2ds\\
&=T_{\beta}^+J_{\beta}\tilde W_{\beta}^+f_2+i\int_0^{\infty}e^{isH}G_{\beta}^+\tilde W_{\beta}^+e^{-is \tilde h_{\beta}}f_2ds,
\end{split}
\end{equation}
where $h_{\beta}:=\tilde H_b+E_{\beta}$, $\tilde W_{\beta}^{\pm}:=s-\lim_{t\to\pm\infty} e^{it\tilde H_b}e^{-iS_{\beta,\pm}(D_b,t)}$, $\tilde h_{\beta}:=-\Delta_b+E_{\beta}$, $T_{b}^+:=Op(t_{b}^+)$, $t_{b}^+:=\varphi(\xi_b)\psi(\lvert\xi_b\rvert^2)\eta(\lvert x_b\rvert)\theta_+(\hat x_b\cdot \hat\xi_b)$, and $G_{\beta}^+:=HT_{b}^+J_{\beta}-T_{b}^+J_{\beta}h_{\beta}$. Here $S_{\beta,\pm}({\xi}_b,t)$ is the function in Lemma \ref{HamiJac} obtained by replacing $V$ by $\tilde I_b$.

Since $\mathrm{supp}\, \varphi\in C_b'$ and $\hat x_b\cdot\hat\xi_b>1-2\epsilon$ on $\mathrm{supp}\, \theta_+(\hat x_b\cdot\hat\xi_b)$, we have $\hat x_b\in C_b'$ on $\mathrm{supp}\, t_b^+$ for sufficiently small $\epsilon$. We also have $\hat x_b\cdot\hat \xi_b<1-\epsilon$ on $\mathrm{supp}\, \theta_+'(\hat x_b\cdot\hat \theta_b)$, where $\theta_+'(t):=\frac{d}{dt}\theta_+(t)$. Therefore, using Lemma \ref{micro1} we can see that if $v\in B^*(X_a)$ is strictly outgoing or incoming, then $G_{\beta}^+v \in L^{2,l}(X)$ for some $l>1/2$.

In the same way we have
$$W_{\alpha}^-f_1=T_{a}J_{\alpha}\tilde W_{\alpha}^-f_1-i\int_{-\infty}^0e^{isH}G_{\alpha}\tilde W_{\alpha}^-e^{-is\tilde h_a}f_1ds,$$
where $T_{a}:=Op(t_{a})$, $t_{a}:=\varphi(\xi_a)\tilde\psi(\lvert\xi_a\rvert^2)\eta(\lvert x_a\rvert)(\theta_+(\hat x_a\cdot \hat\xi_a)+\theta_-(\hat x_a\cdot \hat\xi_a))$, $\theta_-(t):=\theta_+(-t)$, and $G_{\alpha}:=HT_{a}J_{\alpha}-T_{a}J_{\alpha}h_{\alpha}$. Here $\tilde\psi\in C^{\infty}_0(\mathbb R)$ be a function satisfying $\tilde\psi(t)=1$ for any $t>0$ such that $\mathrm{supp}\, \hat f_1(\sqrt{t},\cdot)\neq \varnothing$.

We also have
$$\Omega_{\alpha}^+\tilde W_{\alpha}^-f_1=T_{a}J_{\alpha}\tilde W_{\alpha}^-f_1+i\int_0^{\infty}e^{isH}G_{\alpha}\tilde W_{\alpha}^-e^{-is\tilde h_a}f_1ds,$$
where $\Omega_{\alpha}^{\pm}:=s-\lim_{t\to\pm\infty}e^{itH}J_{\alpha}e^{-ith_{\alpha}}\mathbf{1}_{Z_a}(p_a^+)$, with $Z_a$ defined as \eqref{myeqfirst.2.0.1} (see \cite[Theorem 6.10.1]{DG}). Here $p_a^+$ is the asymptotic velocity for $\tilde H_a$. Note that $\mathrm{Ran}\, \tilde W_{\alpha}^-=\mathbf{1}_{Z_a}(p_a^+)$ (see \cite[Theorem 6.15.2]{DG}).
Thus, we obtain
$$(\Omega_{\alpha}^+\tilde W_{\alpha}^--W_{\alpha}^-)f_1=i\int_{-\infty}^{\infty}e^{isH}G_{\alpha}\tilde W_{\alpha}^-e^{-is\tilde h_a}f_1ds.$$
Since $W_{\alpha}^+=\Omega_{\alpha}^+\tilde W_{\alpha}^+$, we have
$$S_{\alpha\alpha}- S_{\alpha}=(W_{\alpha}^+)^*(W_{\alpha}^--\Omega_{\alpha}^+\tilde W_{\alpha}^-),$$
where $S_{\alpha}:=(\tilde W_{\alpha}^+)^*\tilde W_{\alpha}^-$.
Therefore, by \eqref{myeq6.0.1} we obtain 
\begin{equation}\label{myeq6.1}
\begin{split}
(S_{\alpha\alpha}f_1,&f_2)_{\mathcal H}-(S_{\alpha}f_1,f_2)_{\mathcal H}\\
&=-i\int_{-\infty}^{\infty}dt(e^{itH}G_{\alpha}\tilde W_{\alpha}^-e^{-it\tilde h_{\alpha}}f_1,W_{\alpha}^+f_2)_{\mathcal H}\\
&=-i\int_{-\infty}^{\infty}dt(G_{\alpha}\tilde W_{\alpha}^-e^{-it\tilde h_{\alpha}}f_1,T_{a}^+J_{\alpha}\tilde W_{\alpha}^+e^{-it\tilde h_{\alpha}}f_2)_{\mathcal H}\\
&\quad -\int_0^{\infty}ds\int_{-\infty}^{\infty}dt(G_{\alpha}\tilde W_{\alpha}^-e^{-it\tilde h_{\alpha}}f_1,e^{isH}G_{\alpha}^+\tilde W_{\alpha}^+e^{-i(s+t)\tilde h_{\alpha}}f_2)_{\mathcal H},
\end{split}
\end{equation}
where $\mathcal H:=L^2(X)$.

Here we note $\tilde W_{\alpha}^{\pm}F_{\alpha}^*=\tilde F_{\alpha,\pm}^*$ where $(\tilde F_{\alpha,\pm}f)(\lambda):=(\tilde F_{a,\pm}f)(\lambda-E_{\alpha})$, Here $\tilde F_{a,\pm}$ is the generalized Fourier transform corresponding to $F_{\pm}$  defined in section \ref{secondsec} with $V$ replaced by $\tilde I_a$. Thus, we have
$$\tilde W_{\alpha}^{\pm}e^{-it\tilde h_
\alpha}f_1=\tilde F_{\alpha,\pm}^*e^{-it\lambda}\hat f_1.$$

The second term on the right-hand side of \eqref{myeq6.1} is calculated as
\begin{align*}
&(G_{\alpha}\tilde W_{\alpha}^-e^{-it\tilde h_{\alpha}}f_1,e^{isH}G_{\alpha}^+\tilde W_{\alpha}^+e^{-i(s+t)\tilde h_{\alpha}}f_2)_{\mathcal H}\\
&=(e^{-it\lambda}\hat f_1,\tilde F_{\alpha,-}G_{\alpha}^*e^{isH}G_{\alpha}^+\tilde W_{\alpha}^+e^{-i(s+t)\tilde h_{\alpha}}f_2)_{\hat {\mathcal H}_a}\\
&=\int_{E_{\alpha}}^{\infty}e^{-it\lambda}(\hat f_1(\lambda),\tilde {\mathcal F}_{\alpha,-}(\lambda)G_{\alpha}^*e^{isH}G_{\alpha}^+\tilde W_{\alpha}^+e^{-i(s+t)\tilde h_{\alpha}}f_2)_{K_a}d\lambda\\
&=\int_{E_{\alpha}}^{\infty}e^{-it\lambda}( (G_{\alpha}^+)^*e^{-isH} G_{\alpha}\tilde {\mathcal F}_{\alpha,-}^*(\lambda)\hat f_1(\lambda),\tilde W_{\alpha}^+e^{-i(s+t)\tilde h_{\alpha}}f_2)_{\mathcal H} d\lambda\\
&=\int_{E_{\alpha}}^{\infty}\int_{E_{\alpha}}^{\infty}e^{-it\lambda+i(s+t)\lambda'}\\
&\qquad\cdot (\tilde {\mathcal F}_{\alpha,+}(\lambda') (G_{\alpha}^+)^*e^{-isH} G_{\alpha}\tilde {\mathcal F}_{\alpha,-}^*(\lambda)\hat f_1(\lambda),\hat f_2(\lambda'))_{K_a} d\lambda d\lambda'\\
&=\int_{E_{\alpha}}^{\infty}\int_{E_{\alpha}}^{\infty}e^{-it\lambda+i(s+t)\lambda'}(e^{-isH} G_{\alpha}\tilde {\mathcal F}_{\alpha,-}^*(\lambda)\hat f_1(\lambda), G_{\alpha}^+\tilde {\mathcal F}_{\alpha,+}^*(\lambda')\hat f_2(\lambda'))_{\mathcal H} d\lambda d\lambda'.
\end{align*}
Here $\hat {\mathcal H_a}:=L^2((E_{\alpha},\infty); L^2(C_a))$, and $\tilde {\mathcal F}_{\alpha,\pm}(\lambda):=\tilde {\mathcal F}_{a,\pm}(\lambda-E_{\alpha})$, where $\tilde {\mathcal F}_{a,\pm}(\lambda-E_{\alpha})$ is the operator corresponding to $\mathcal F_{\pm}(\lambda)$ in section \ref{secondsec} with $V$ replaced by $\tilde I_a$.

Inserting the convergent factor $e^{-\epsilon s}$, integrating with respect to $s$ and taking the limit as $\epsilon \to 0_+$ we obtain
$$-i\int_{E_{\alpha}}^{\infty}\int_{E_{\alpha}}^{\infty}e^{-it(\lambda-\lambda')}(R(\lambda'+i0) G_{\alpha}\tilde {\mathcal F}_{\alpha,-}^*(\lambda)\hat f_1(\lambda), G_{\alpha}^+\tilde {\mathcal F}_{\alpha,+}^*(\lambda')\hat f_2(\lambda'))  d\lambda d\lambda',$$
where $(v_1,v_2):=\int_Xv_1(x)\bar v_2(x)dx$.

We again insert the factor $e^{-\epsilon\lvert t\rvert}$, integrate with respect to $t$, and take the limit as $\epsilon\to0_+$. Then, the second term on the right-hand side of \eqref{myeq6.1} is written as
\begin{equation*}\label{myeq6.1.1.1}
2\pi i\int_{E_{\alpha}}^{\infty}(R(\lambda+i0) G_{\alpha}\tilde {\mathcal F}_{\alpha,-}^*(\lambda)\hat f_1(\lambda), G_{\alpha}^+\tilde {\mathcal F}_{\alpha,+}^*(\lambda)\hat f_2(\lambda)) d\lambda.
\end{equation*}

In the similar way the first term on the right-hand side of \eqref{myeq6.1} is written as
\begin{equation}\label{myeq6.1.1}
-2\pi i\int_{E_{\alpha}}^{\infty}(G_{\alpha}\tilde {\mathcal F}_{\alpha,-}^*(\lambda)\hat f_1(\lambda), T_{a}^+J_{\alpha}\tilde {\mathcal F}_{\alpha,+}^*(\lambda)\hat f_2(\lambda)) d\lambda.
\end{equation}

Noting $(\tilde H_a-\lambda_{\alpha})\tilde {\mathcal F}_{\alpha,-}^*(\lambda)=0$ we obtain
\begin{equation}\label{myeq6.1.1.2}
\begin{split}
&G_{\alpha}^+\tilde {\mathcal F}_{\alpha,-}^*(\lambda)\\
&\quad=\{(H-\lambda)T_a^+J_{\alpha}-T_a^+J_{\alpha}(\tilde H_a-\lambda_{\alpha})\}\tilde {\mathcal F}_{\alpha,-}^*(\lambda)=(H-\lambda)T_a^+J_{\alpha}\tilde {\mathcal F}_{\alpha,-}^*(\lambda).
\end{split}
\end{equation}

Therefore, we have
\begin{equation}\label{myeq6.1.2}
\begin{split}
&(S_{\alpha\alpha}f_1,f_2)-(S_{\alpha}f_1,f_2)\\
&\quad=2\pi i\int_{E_{\alpha}}^{\infty}(G_{\alpha}\tilde {\mathcal F}_{\alpha,-}^*(\lambda)\hat f_1(\lambda), R(\lambda-i0)(H-\lambda)T_{a}^+J_{\alpha}\tilde {\mathcal F}_{\alpha,+}^*(\lambda)\hat f_2(\lambda)) d\lambda\\
&\qquad-2\pi i\int_{E_{\alpha}}^{\infty}(G_{\alpha}\tilde {\mathcal F}_{\alpha,-}^*(\lambda)\hat f_1(\lambda), T_{a}^+J_{\alpha}\tilde {\mathcal F}_{\alpha,+}^*(\lambda)\hat f_2(\lambda)) d\lambda.
\end{split}
\end{equation}

We can replace $T_a^+\tilde {\mathcal F}_{\alpha,+}^*(\lambda)\hat f_2(\lambda)$ in \eqref{myeq6.1.2} by $-C_{\alpha}^+(\lambda)T_a^+w_{\alpha}^+(\lambda)$ where
$$(w_{\alpha}^+(\lambda))(x_a)=w_{\alpha}^+(\lambda,x_a):=\eta(r_a)\hat f_2(\lambda,\hat x_a)e^{iK_{a}(x_a,\lambda_{\alpha})}r_a^{-(n_a-1)/2},$$
and
$$C_{\alpha}^+(\lambda):=2^{-1}i\pi^{-1/2}\lambda_{\alpha}^{-1/4}e^{-i\pi(n_a-3)/4},$$
with $\eta$ being as in \eqref{myeq2.3} for some $\kappa$.

To see that, set $F_1^+:=J_{\alpha}T_{a}^+\tilde R_a(\lambda_{\alpha}-i0)\tilde w_{\alpha}^+(\lambda)$, where $\tilde H_a:=-\Delta_a+\tilde I_a$, $\tilde w_{\alpha}^+(\lambda):=(\tilde H_a-\lambda_{\alpha})w_{\alpha}^+(\lambda)$ and $\tilde R_a(\lambda_{\alpha}-i0):=(\tilde H_a-\lambda+i0)^{-1}$. Then, by Lemma \ref{outdirect}, $F_1$ is incoming, and therefore, by Lemma \ref{Iunique} we have
$$F_1^+-R(\lambda-i0)F_1^+=0.$$
Since we have
$$\tilde {\mathcal F}_{\alpha,+}^*(\lambda)\hat f_2(\lambda)=-C_{\alpha}^+(\lambda)(w_{\alpha}^+(\lambda)-F_1^+),$$
$T_a^+\tilde {\mathcal F}_{\alpha,+}^*(\lambda)\hat f_2(\lambda)$ in \eqref{myeq6.1.2} can be replaced by $-C_{\alpha}^+(\lambda)T_a^+w_{\alpha}^+(\lambda)$.

Set $\tilde T_a^+:=1-T_a^+$. Since by Lemma \ref{micro2} we have $\tilde T_{a}^+w_{\alpha}^+(\lambda)\in \mathcal S(X_a)$, by Lemma \ref{outdirect} $F_2^+:=J_{\alpha}\tilde T_a^+w_{\alpha}^+(\lambda)$ is incoming, and therefore
$$F_2^+-R(\lambda-i0)F_2^+=0.$$

Thus, we can remove $T_a^+$ in front of $w_{\alpha}^+(\lambda)$. Therefore, we can rewrite \eqref{myeq6.1.2} as
\begin{equation}\label{myeq6.1.3}
\begin{split}
&(S_{\alpha\alpha}f_1,f_2)-(S_{\alpha}f_1,f_2)\\
&\quad=-\tilde C_{\alpha}(\lambda)\int_{E_{\alpha}}^{\infty}(R(\lambda+i0) G_{\alpha}\tilde {\mathcal F}_{\alpha,-}^*(\lambda)\hat f_1(\lambda), (H-\lambda)J_{\alpha}w_{\alpha}^+(\lambda)) d\lambda\\
&\qquad+\tilde C_{\alpha}(\lambda)\int_{E_{\alpha}}^{\infty}(G_{\alpha}\tilde {\mathcal F}_{\alpha,-}^*(\lambda)\hat f_1(\lambda), J_{\alpha}w_{\alpha}^+(\lambda)) d\lambda,
\end{split}
\end{equation}
where $\tilde C_{\alpha}(\lambda):=\pi^{1/2}\lambda_{\alpha}^{-1/4}e^{i\pi(n_a-3)/4}$.

By Lemma \ref{2bodyasympt} and \eqref{myeq5.9} we can see
\begin{align*}
(\hat S_{\alpha}(\lambda)\hat f_1(\lambda),\hat f_2(\lambda))&=\tilde C_{\alpha}(\lambda)\{(T_{a}\tilde {\mathcal F}_{\alpha,-}^*(\lambda)\hat f_1(\lambda),\tilde w_{\alpha}^+(\lambda))_a\\
&\qquad-((\tilde H_a-\lambda_{\alpha})T_{a}\tilde {\mathcal F}_{\alpha,-}^*(\lambda)\hat f_1(\lambda),w_{\alpha}^+(\lambda))_a\},
\end{align*}
where $(v_1,v_2)_a:=\int_{X_a}v_1(x_a)v_2(x_a)dx_a$ and $\hat S_{\alpha}(\lambda)$ is the fiber of $F_{\alpha} S_{\alpha}F_{\alpha}^*$ (see \cite{RS}).

Therefore, we obtain
\begin{equation}\label{myeq6.1.3.1}
\begin{split}
(S_{\alpha} f_1,f_2)=&\int_{E_{\alpha}}^{\infty}\tilde C_{\alpha}(\lambda)\{(T_{a}\tilde {\mathcal F}_{\alpha,-}^*(\lambda)\hat f_1(\lambda),\tilde w_{\alpha}^+(\lambda))_a\\
&\quad-((\tilde H_a-\lambda_{\alpha})T_{a}\tilde {\mathcal F}_{\alpha,-}^*(\lambda)\hat f_1(\lambda),w_{\alpha}^+(\lambda))_a\}d\lambda\\
=&\int_{E_{\alpha}}^{\infty}\tilde C_{\alpha}(\lambda)\{(J_{\alpha}T_{a}\tilde {\mathcal F}_{\alpha,-}^*(\lambda)\hat f_1(\lambda),J_{\alpha}\tilde w_{\alpha}^+(\lambda))\\
&\quad-((\tilde H_a-\lambda_{\alpha})J_{\alpha}T_{a}\tilde {\mathcal F}_{\alpha,-}^*(\lambda)\hat f_1(\lambda),J_{\alpha}w_{\alpha}^+(\lambda))\}d\lambda\\
=&\int_{E_{\alpha}}^{\infty}\tilde C_{\alpha}(\lambda)\{(J_{\alpha}T_{a}\tilde {\mathcal F}_{\alpha,-}^*(\lambda)\hat f_1(\lambda),(H-\lambda)J_{\alpha}w_{\alpha}^+(\lambda))\\
&\quad-((H-\lambda)J_{\alpha}T_{a}\tilde {\mathcal F}_{\alpha,-}^*(\lambda)\hat f_1(\lambda),J_{\alpha}w_{\alpha}^+(\lambda))\}d\lambda.
\end{split}
\end{equation}

In the same way as in \eqref{myeq6.1.1.2} we have
\begin{equation}\label{myeq6.1.3.2}
G_{\alpha}\tilde {\mathcal F}_{\alpha,-}^*(\lambda)=(H-\lambda)T_aJ_{\alpha}\tilde {\mathcal F}_{\alpha,-}^*(\lambda).
\end{equation}

Thus, by \eqref{myeq6.1.3} and \eqref{myeq6.1.3.1} we have
\begin{align*}
&(S_{\alpha\alpha} f_1,f_2)\\
&\quad=\tilde C_{\alpha}(\lambda)\int_{E_{\alpha}}^{\infty}\{(J_{\alpha}T_{a}\tilde {\mathcal F}_{\alpha,-}^*(\lambda)\hat f_1(\lambda),(H-\lambda)(J_{\alpha}w_{\alpha}^+(\lambda)))\\
&\qquad-(R(\lambda+i0) G_{\alpha}\tilde {\mathcal F}_{\alpha,-}^*(\lambda)\hat f_1(\lambda), (H-\lambda)(J_{\alpha}w_{\alpha}^+(\lambda)))\}d\lambda.\\
&\quad=\tilde C_{\alpha}(\lambda)\int_{E_{\alpha}}^{\infty}\{(J_{\alpha}T_{a}\tilde {\mathcal F}_{\alpha,-}^*(\lambda)\hat f_1(\lambda),(H-\lambda)(J_{\alpha}w_{\alpha}^+(\lambda)))\\
&\qquad-(R(\lambda+i0) (H-\lambda)J_{\alpha}T_a\tilde {\mathcal F}_{\alpha,-}^*(\lambda)\hat f_1(\lambda), (H-\lambda)(J_{\alpha}w_{\alpha}^+(\lambda)))\}d\lambda.
\end{align*}

By Lemma \ref{outdirect} $F_1^-:=J_{\alpha}T_a\tilde R_a(\lambda_{\alpha}+i0)(\tilde H_a-\lambda_{\alpha})w_{\alpha}^-(\lambda)$ is outgoing, where $(w_{\alpha}^-(\lambda))(x_a)=w_{\alpha}^-(\lambda,x_a):=\eta(r_a)\hat f_1(\lambda,-\hat x_a)e^{-iK_a(x_a,\lambda_{\alpha})}r_a^{-(n_a-1)/2}$. Thus, we have
\begin{equation}\label{myeq6.1.4}
F_1^--R(\lambda+i0)(H-\lambda)F_1^-=0.
\end{equation}

Set $\tilde T_a:=1-T_a$. Then, by Lemma \ref{outdirect} $F_2^-:=J_{\alpha}\tilde T_aw_{\alpha}^-(x_a,\lambda_{\alpha})$ is outgoing. Thus, we have
\begin{equation}\label{myeq6.1.5}
F_2^--R(\lambda+i0)(H-\lambda)F_2^-=0.
\end{equation}

By \eqref{myeq6.1.4} and \eqref{myeq6.1.5} we can replace $T_a\tilde {\mathcal F}^*_{\alpha,-}(\lambda)\hat f_1(\lambda)$ by $C_{\alpha}^-(\lambda)w_{\alpha}^-(\lambda)$, where $C_{\alpha}^-(\lambda):=2^{-1}i\pi^{-1/2}\lambda_{\alpha}^{-1/4}e^{i\pi(n_a-3)/4}$.
Thus, by \eqref{myeq5.9} we obtain
\begin{align*}
(S_{\alpha\alpha}f_1,f_2)=&(-2^{-1}i\lambda_{\alpha}^{-1/2})^{-1}C_{\alpha}^-(\lambda)\tilde C_{\alpha}(\lambda)\\
& \cdot\int_{E_{\alpha}}^{\infty}(\Sigma_{\alpha\alpha}(\lambda)\mathcal R\hat f_1(\lambda),\hat f_2(\lambda))d\lambda\\
=&e^{i\pi(n_a-1)/2}\int_{E_{\alpha}}^{\infty}(\Sigma_{\alpha\alpha}(\lambda)\mathcal R\hat f_1(\lambda),\hat f_2(\lambda))d\lambda,
\end{align*}
and therefore, \eqref{myeq6.0.0.1} for $\alpha=\beta$.

(ii) Next we consider $S_{\beta\alpha},\ \beta\neq\alpha$. Since $(\Omega_{\beta}^+)^*\Omega_{\alpha}^+=0$(see e.g. \cite[Theorem 6.15.3]{DG}), we have
$$S_{\beta\alpha}=(W_{\beta}^+)^*(W_{\alpha}^--\Omega_{\alpha}^+\tilde W_{\alpha}^-).$$
Hence, in the same way as above we obtain
\begin{equation}\label{myeq6.4.1}
\begin{split}
&(S_{\beta\alpha}f_1,f_2)\\
&=2\pi i\int_{\max\{E_{\alpha},E_{\beta}\}}^{\infty}\{(R(\lambda+i0)G_{\alpha}\tilde {\mathcal F}_{\alpha,-}^*(\lambda)\hat f_1(\lambda),G_{\beta}^+\tilde {\mathcal F}_{\beta,+}^*(\lambda)\hat f_2(\lambda))\\
&\quad-(G_{\alpha}\tilde {\mathcal F}_{\alpha,-}^*(\lambda)\hat f_1(\lambda),T_{b}^+J_{\beta}\tilde {\mathcal F}_{\beta,+}^*(\lambda)\hat f_2(\lambda))\}d\lambda\\
&=2\pi i\int_{\max\{E_{\alpha},E_{\beta}\}}^{\infty}\{(R(\lambda+i0)G_{\alpha}\tilde {\mathcal F}_{\alpha,-}^*(\lambda)\hat f_1(\lambda),(H-\lambda)T_{b}^+J_{\beta}\tilde {\mathcal F}_{\beta,+}^*(\lambda)\hat f_2(\lambda))\\
&\quad-(G_{\alpha}\tilde {\mathcal F}_{\alpha,-}^*(\lambda)\hat f_1(\lambda),T_{b}^+J_{\beta}\tilde {\mathcal F}_{\beta,+}^*(\lambda)\hat f_2(\lambda))\}d\lambda.
\end{split}
\end{equation}

We can replace $T_b^+\tilde {\mathcal F}_{\beta,+}^*(\lambda)\hat f_2(\lambda)$ in \eqref{myeq6.4.1} by $-C_{\beta}(\lambda)w_{\beta}^+(\lambda)$ where
$$C_{\beta}^+(\lambda):=2^{-1}i\pi^{-1/2}\lambda_{\beta}^{-1/4}e^{-i\pi(n_b-3)/4},$$
and $(w_{\beta}^+(\lambda))(x_b)=w_{\beta}^+(\lambda,x_b):=\eta(r_b)\hat f_2(\lambda,\hat x_b)e^{iK_{b}(x_b,\lambda_{\beta})}r_b^{-(n_b-1)/2}$.

We shall prove
\begin{equation}\label{myeq6.5}
\begin{split}
(G_{\alpha}\tilde {\mathcal F}_{\alpha,-}^*(\lambda)\hat f_1(\lambda)&,J_{\beta}w_{\beta}^+(\lambda))\\
&=((H-\lambda)T_aJ_{\alpha}\tilde {\mathcal F}_{\alpha,-}^*(\lambda)\hat f_1(\lambda),J_{\beta}w_{\beta}^+(\lambda))\\
&=(T_{a}J_{\alpha}\tilde {\mathcal F}_{\alpha,-}^*(\lambda)\hat f_1(\lambda),(H-\lambda)J_{\beta}w_{\beta}^+(\lambda)).
\end{split}
\end{equation}

When $a=b$ and $\alpha\neq\beta$, by \eqref{myeq6.1.3.2}, $\pi_{\beta}J_{\alpha}=0$ and $\pi_{\beta}(-\Delta_a)J_{\alpha}=0$, \eqref{myeq6.5} holds. 

When $a\neq b$, (a) $J_{\alpha}\tilde {\mathcal F}_{\alpha,-}^*(\lambda)\hat f_1(\lambda)$ exponentially decays on $Y_b^{\epsilon}$ for any $\epsilon>0$, or (b) $J_{\beta}\tilde {\mathcal F}_{\beta,-}^*(\lambda)\hat f_2(\lambda)$ exponentially decays on $Y_a^{\epsilon}$ for any $\epsilon>0$, where $Y_a^{\epsilon}$ is defined as \eqref{myeqfirst.2.0.2}.

In the case (a) let $\chi_j\in C^{\infty}(X),\ j=1,2$ be functions satisfying the following: there exists $\epsilon'>0$ such that $\chi_1=1$ on $\bigcup_{c\nleq b}\{x: \lvert x\rvert >1,\ 2\epsilon'\lvert x\rvert>\lvert x^c\rvert\}$, $\chi_1=0$ on $\bigcup_{c\nleq b}\{x: \lvert x\rvert >1,\ \epsilon'\lvert x\rvert<\lvert x^c\rvert\}$, and $\chi_1+\chi_2=1$. We also denote by $\chi_j$ the multiplication operator by $\chi_j$.

Then, for $\epsilon'$ sufficiently small $\chi_1J_{\beta}\tilde {\mathcal F}_{\beta,-}^*(\lambda)\hat f_2(\lambda)$ decays exponentially, and $J_{\alpha}\tilde {\mathcal F}_{\alpha,-}^*(\lambda)\hat f_1(\lambda)$ decays exponentially on $\mathrm{supp}\, \chi_2$. Thus we can see that
\begin{align*}
(\Delta T_{a}&J_{\alpha}\tilde {\mathcal F}_{\alpha,-}^*(\lambda)\hat f_1(\lambda),\chi_jT_{b}^+J_{\beta}\tilde {\mathcal F}_{\beta,+}^*(\lambda)\hat f_2(\lambda))\\
&=(T_{a}J_{\alpha}\tilde {\mathcal F}_{\alpha,-}^*(\lambda)\hat f_1(\lambda),\Delta \chi_jT_{b}^+J_{\beta}\tilde {\mathcal F}_{\beta,+}^*(\lambda)\hat f_2(\lambda)),\ j=1,2
\end{align*}
and therefore, \eqref{myeq6.5} holds.

In the same way we can see that in the case (b) \eqref{myeq6.5} holds.

Thus, by \eqref{myeq6.4.1} we obtain
\begin{equation}\label{myeq6.6}
\begin{split}
(S_{\beta\alpha}f_1,f_2)=&\tilde C_{\beta}(\lambda)\int_{\max\{E_{\alpha},E_{\beta}\}}^{\infty}\{(G_{\alpha}\tilde {\mathcal F}_{\alpha,-}^*(\lambda)\hat f_1(\lambda),(H-\lambda)J_{\beta}w_{\beta}^+(\lambda))\\
&-(R(\lambda+i0)G_{\alpha}\tilde {\mathcal F}_{\alpha,-}^*(\lambda)\hat f_1(\lambda),(H-\lambda)J_{\beta}w_{\beta}^+(\lambda))\}d\lambda.
\end{split}
\end{equation}

As in the case of $\hat S_{\alpha\alpha}(\lambda)$, we can replace $T_a\tilde {\mathcal F}^*_{\alpha,-}(\lambda)\hat f_1(\lambda)$ by $C_{\alpha}^-(\lambda)w_{\alpha}^-(\lambda)$, and therefore, we obtain \eqref{myeq6.0.0.1}.
\end{proof}

\section{Equivalence and adjoint operators of generalized Fourier transform}\label{eighthsec}
In the following we use the notations in section \ref{firstsec}, \ref{first.2sec}, \ref{secondsec} and \ref{sixthsec}.

\begin{proof}[Proof of Theorem \ref{gfe} (1)]
Let $f\in L^{2,l}(X),\ l>1/2$ and $f_1\in C^{\infty}(X_a)$ satisfy $\hat f_1\in C_0^{\infty}(X_a)$ and $\hat f_1(\lambda,\cdot)\in C_0^{\infty}(C_a')$ for any $\lambda$ such that $\mathrm{supp}\, \hat f_1(\lambda,\cdot)\neq 0$. Here $(\hat f_1(\lambda))(\hat x_a)=\hat f_1(\lambda,\hat x_a):=(F_{\alpha}f_1)(\lambda,\hat x_a)$.

Then, we have
$$(F_{\alpha}(W_{\alpha}^+)^*f,\hat f_1)=(f_1,W_{\alpha}^+(F_{\alpha})^*\hat f_1).$$

By \eqref{myeq6.0.1} the right-hand side is written as
\begin{equation}\label{myeq8.1}
(f,T_{\alpha}^+J_{\alpha}\tilde W_{\alpha}^+(F_{\alpha})^*\hat f_1)-i\int_0^{\infty}(f,e^{isH}G_{\alpha}^+\tilde W_{\alpha}^+e^{-is\tilde h_{\alpha}}(F_{\alpha})^*\hat f_1)ds.
\end{equation}
As in the proof of Theorem \ref{sme} we can rewrite this as
\begin{align*}
\int_{E_{\alpha}}^{\infty}&\{(f,T_{a}^+J_{\alpha}\tilde {\mathcal F}_{\alpha,+}^*(\lambda)\hat f_1(\lambda))-(f,R(\lambda-i0)G_{\alpha}^+\tilde {\mathcal F}_{\alpha,+}^*(\lambda)\hat f_1(\lambda))\}d\lambda,\\
&=\int_{E_{\alpha}}^{\infty}\{(f,T_{a}^+J_{\alpha}\tilde {\mathcal F}_{\alpha,+}^*(\lambda)\hat f_1(\lambda))\\
&\quad-(f,R(\lambda-i0)(H-\lambda)T_{a}^+J_{\alpha}\tilde {\mathcal F}_{\alpha,+}^*(\lambda)\hat f_1(\lambda))\}d\lambda.
\end{align*}

As in the proof of Theorem \ref{sme} we can replace $T_a^+\tilde {\mathcal F}_{\alpha,+}^*(\lambda)\hat f_1(\lambda)$ by $-C_{\alpha}^+(\lambda)\hat w_{\alpha}^+(\lambda)$, where
$$(\hat w_{\alpha}^+(\lambda))(x_a)=\hat w_{\alpha}^+(\lambda,x_a):=\eta(r_a)\hat f_1(\lambda,\hat x_a)e^{iK_a(x_a,\lambda_{\alpha})}r_a^{-(n_a-1)/2}.$$

Thus, we obtain
\begin{equation}\label{myeq8.2}
\begin{split}
(F_{\alpha}(W_{\alpha}^+)^*f,\hat f_1)&=\hat C_{\alpha}^+(\lambda)\int_{E_{\alpha}}^{\infty}\{(f_1,J_{\alpha}w_{\alpha}^+(\lambda))\\
&\quad-(f_1,R(\lambda-i0)(H-\lambda)J_{\alpha}w_{\alpha}^+(\lambda))\}d\lambda,
\end{split}
\end{equation}
where $\hat C_{\alpha}^+(\lambda):=2^{-1}i\pi^{-1/2}\lambda_{\alpha}^{-1/4}e^{i\pi(n_a-3)/4}$.

Therefore, by \eqref{myeq5.10} we have $(F_{\alpha}(W_{\alpha}^+)^*f)(\lambda)=\mathcal G_{\alpha}^+(\lambda)f.$

In the same way we obtain
\begin{equation}\label{myeq8.3}
\begin{split}
(F_{\alpha}(W_{\alpha}^-)^*f,\hat f_1)&=\hat C_{\alpha}^-(\lambda)\int_{E_{\alpha}}^{\infty}\{(f,J_{\alpha}\hat w_{\alpha}^-(\lambda))\\
&\quad-(f,R(\lambda+i0)(H-\lambda)J_{\alpha}\hat w_{\alpha}^-(\lambda))\}d\lambda,
\end{split}
\end{equation}
where $(w_{\alpha}^-(\lambda))(x_a)=w_{\alpha}^-(\lambda,x_a):=\eta(r_a)f_1(\lambda,-\hat x_a)e^{-iK_a(x_a,\lambda_{\alpha})}r_a^{-(n_a-1)/2}$ and
$$\hat C_{\alpha}^-(\lambda):=2^{-1}i\pi^{-1/2}\lambda_{\alpha}^{-1/4}e^{-i\pi(n_a-3)/4}.$$

Therefore, by \eqref{myeq5.10} we have $(F_{\alpha}(W_{\alpha}^-)^*f)(\lambda)=\mathcal G_{\alpha}^-(\lambda)f.$

\end{proof}

\begin{proof}[Proof of Theorem \ref{gfe} (2)]
By \eqref{myeq8.2} and Theorem \ref{gfe} (1) we can see that
\begin{equation}\label{myeq8.4}
(\mathcal G_{\alpha}^+(\lambda))^*=\overline{\hat C_{\alpha}^+(\lambda)}P_{\alpha,-}(\lambda).
\end{equation}

In the same way by \eqref{myeq8.3} and Theorem \ref{gfe} (1) we can see that
\begin{equation}\label{myeq8.5}
(\mathcal G_{\alpha}^-(\lambda))^*=\overline{\hat C_{\alpha}^-(\lambda)}P_{\alpha,+}(\lambda)\mathcal R.
\end{equation}

Theorem \ref{gfe} (2) follows from \eqref{myeq8.4} and \eqref{myeq8.5}.
\end{proof}

\end{document}